\documentclass[a4paper,11pt]{article} 
\usepackage{fancyhdr}
\usepackage{graphicx}
\usepackage{microtype} 

\usepackage[backend=bibtex,citestyle=alphabetic,bibstyle=alphabetic,backref=true,url=false]{biblatex}
\usepackage{amsfonts,amssymb,amstext,amsmath} 
\usepackage{cool} 
\usepackage[amsmath,thmmarks,hyperref]{ntheorem}
\usepackage{tensor}
\usepackage{tensind} 
\tensordelimiter{?}

\usepackage[bookmarksnumbered,hypertexnames=false]{hyperref} 
\hypersetup{
    plainpages=false,       
    unicode=false,          
    pdftoolbar=true,        
    pdfmenubar=true,        
    pdffitwindow=false,     
    pdfstartview={FitH},    
    pdftitle={Killing tensors, Warped Products and The Orthogonal Separation of The Hamilton-Jacobi Equation},    
    pdfauthor={Krishan Rajaratnam, Raymond G. McLenaghan},     
    pdfkeywords= {hamilton-jacobi equation} {completely integrable systems} {separation of variables} {Killing tensor} {warped product} {Stackel systems} {special conformal killing tensor} {concircular tensor}, 
    pdfnewwindow=true,      
    linktoc=page,           
    colorlinks=true,        
    linkcolor=blue,         
    citecolor=black,        
    filecolor=magenta,      
    urlcolor=cyan           
}

\usepackage[capitalize]{cleveref} 
\usepackage{amsmath,amsfonts,amssymb}

\newcommand{\spa}[1]{\operatorname{span} \{#1\}}
\newcommand{\End}{\operatorname{End}}

\newcommand{\tr}[1]{\operatorname{tr}(#1)}

\def\sp(#1,#2){\left\langle #1,#2 \right\rangle}
\def\bp#1{\sp(#1)}

\DeclareFontFamily{U}{matha}{\hyphenchar\font45}
\DeclareFontShape{U}{matha}{m}{n}{
      <5> <6> <7> <8> <9> <10> gen * matha
      <10.95> matha10 <12> <14.4> <17.28> <20.74> <24.88> matha12
      }{}
\DeclareSymbolFont{matha}{U}{matha}{m}{n}
\DeclareFontFamily{U}{mathx}{\hyphenchar\font45}
\DeclareFontShape{U}{mathx}{m}{n}{
      <5> <6> <7> <8> <9> <10>
      <10.95> <12> <14.4> <17.28> <20.74> <24.88>
      mathx10
      }{}
\DeclareSymbolFont{mathx}{U}{mathx}{m}{n}

\DeclareMathSymbol{\obot}         {2}{matha}{"6B}
\DeclareMathSymbol{\bigobot}       {1}{mathx}{"CB}

\newcommand{\R}{\mathbb{R}}

\renewcommand{\E}{\mathbb{E}} 
\newcommand{\Si}{\mathbb{S}}

\newcommand{\eunn}{\mathbb{E}^{n}_{\nu}}

\newcommand{\sgn}{\operatorname{sgn}}

\def\pderiv#1#2{\dfrac{\partial #1}{\partial #2}}

\newcommand{\lied}[2]{\mathcal{L}_{#2} #1} 
\def\ve{\mathfrak{X}} 

\newcommand{\Ei}{{\mathcal{E}}}
\newcommand{\F}{{\mathcal{F}}}
\def\K{\mathcal{K}}

\usepackage{amsfonts,amssymb,amsmath}

\theoremnumbering{arabic}
\theoremstyle{break} 
\usepackage{latexsym}
\theoremsymbol{\ensuremath{_\Box}}
\theorembodyfont{\itshape}
\theoremheaderfont{\normalfont\bfseries}
\theoremseparator{}
\newtheorem{theorem}{Theorem}[section]
\newtheorem{proposition}[theorem]{Proposition}
\newtheorem{corollary}[theorem]{Corollary}
\newtheorem{lemma}[theorem]{Lemma}
\theoremsymbol{\ensuremath{\diamondsuit}} 
\newtheorem{thmMy}[theorem]{Theorem}
\newtheorem{propMy}[theorem]{Proposition}
\newtheorem{corMy}[theorem]{Corollary}
\theoremsymbol{\ensuremath{_\Box}}

\theorembodyfont{\upshape}
\newtheorem{example}[theorem]{Example}

\newtheorem{remark}[theorem]{Remark}
\newtheorem{definition}[theorem]{Definition}

\theoremstyle{nonumberplain}
\theoremheaderfont{\scshape}
\theorembodyfont{\normalfont}
\theoremsymbol{\ensuremath{_\blacksquare}}
\usepackage{amssymb}
\newtheorem{proof}{Proof}

\qedsymbol{\ensuremath{_\blacksquare}}

\numberwithin{equation}{section} 

\newcounter{partCounter}
\newenvironment{parts}{
	\begin{list}{\bfseries{}Case \arabic{partCounter}~}{\usecounter{partCounter}}
}{
	\end{list}
}
\renewcommand{\d}{{\textrm d}}

\usepackage{autonum} 

\usepackage[sanitize={symbol=false},style=index,sort=def,acronym,toc]{glossaries} 
\newglossary[nlg]{notation}{not}{ntn}{List of Notations} 

\newglossaryentry{tenRCurv}{name=Riemann curvature tensor, description={
$	R(X,Y)Z = \nabla_{X}\nabla_{Y} - \nabla_{Y}\nabla_{X} - \nabla_{[X,Y]}$}}

\newglossaryentry{tenHess}{name=Hessian of a function f, description={
\begin{equation*}
	H^{f}(X,Y) = XY f - (\nabla_{X} Y) f
\end{equation*}
}}

\newglossaryentry{meanCVf}{name=Mean curvature vector field, description={
\begin{equation*}
	H = \frac{1}{n} \sum\limits_{i=1}^{n} \epsilon_{i} \mathbf{II}
\end{equation*}
}}

\newglossaryentry{ot}{name={orthogonal tensor},
description={A valence 2 symmetric contravariant tensor whose uniquely determined endomorphism is point-wise diagonalizable.}}

\newglossaryentry{nTorLess}{name={torsionless},
description={A $\binom{1}{1}$-tensor is called torsionless if it's Nijenhuis torsion vanishes.}}

\newglossaryentry{eunn}{type=notation,name={\ensuremath{\eunn}}, description={pseudo-Euclidean space, an $n$-dimensional vector space equipped with a metric with signature $\nu$}} 

\newglossaryentry{eunnk}{type=notation,name={\ensuremath{\eunn(\kappa)}}, description={ A hyperquadric of pseudo-Euclidean space. More precisely the central hyperquadric of \gls{eunn} with curvature $\kappa$.}}

\newglossaryentry{obot}{type=notation,
name={\ensuremath{\obot}},
description={The orthogonal direct sum.
}}

\newglossaryentry{sgn}{type=notation,
name={\ensuremath{\sgn}},
description={Given a real number $a$, $\sgn a$ is the sign of $a$ if $a \neq 0$ and $0$ if $a = 0$.
}}

\newglossaryentry{fm}{type=notation,
name={\ensuremath{\F(M)}},
description={The set of functions defined on the manifold $M$.
}}

\newglossaryentry{vem}{type=notation,
name={\ensuremath{\ve(M)}},
description={The set of vector fields defined on the manifold $M$.
}}

\newglossaryentry{secd}{type=notation,
name={\ensuremath{\Gamma(E)}},
description={The set of vector fields tangent to the distribution $E$.
}}

\newglossaryentry{spm}{type=notation,
name={\ensuremath{S^{p}(M)}},
description={The set of symmetric contravariant tensors of valence $p$ defined on the manifold $M$.
}}

\newglossaryentry{con}{type=notation,
name={\ensuremath{\operatorname{C}^{p}(M)}},
description={The vector space of concircular contravariant tensors of valence $p$ defined on the manifold $M$.
}}

\newglossaryentry{cono}{type=notation,
name={\ensuremath{\operatorname{C}^p_0(M)}},
description={The vector space of covariantly constant contravariant tensors of valence $p$ defined on the manifold $M$.
}}

\newglossaryentry{odot}{type=notation,
name={\ensuremath{\odot}},
description={The symmetric product of two tensors, i.e. if $u,v$ are tensors then $u \odot v$ is the symmetrization of $u \otimes v$.},
sort=od}

\newglossaryentry{end}{type=notation, name={\ensuremath{\End}},
description={We refer to a linear map from a vector space (vector bundle) into itself as an \emph{endomorphism}.}}

\newacronym{scc}{SCC}{space of constant curvature}

\newacronym{kv}{KV}{Killing vector}
\newacronym{ckv}{CKV}{conformal Killing vector}
\newacronym{kt}{KT}{Killing tensor}
\newacronym{chkt}{ChKT}{characteristic Killing tensor}
\newacronym{ckt}{CKT}{conformal Killing tensor}
\newacronym{cv}{CV}{concircular vector}
\newacronym{ct}{CT}{concircular tensor also called a C-tensor}
\newacronym{oct}{OCT}{orthogonal concircular tensor}
\newacronym{ict}{ICT}{irreducible concircular tensor}

\newacronym{kss}{KS-space}{Killing-Stackel space}
\newacronym{kem}{KEM}{Kalnins-Eisenhart-Miller}
\newacronym{bekm}{BEKM}{Benenti-Eisenhart-Kalnins-Miller}
\newacronym{kbd}{KBD}{Killing Bertrand-Darboux}
\newacronym{kbdt}{KBDT}{Killing Bertrand-Darboux tensor}

\newacronym{wpnet}{WP-net}{warped product net}
\newacronym{tpnet}{TP-net}{twisted product net}

\newacronym{hj}{HJ}{Hamilton-Jacobi}
\makeglossary
\newcommand{\defn}[1]{\emph{\gls*{#1}}} 

\begin{document}
\pagenumbering{roman}
\title{Killing tensors, Warped Products and The Orthogonal Separation of The Hamilton-Jacobi Equation}
\author{Krishan Rajaratnam\footnote{e-mail: k2rajara@uwaterloo.ca}, Raymond G. McLenaghan\footnote{e-mail: rgmclenaghan@uwaterloo.ca}\\Department of Applied Mathematics, University of Waterloo, Canada}
\maketitle
\begin{abstract}\centering
We study Killing tensors in the context of warped products and apply the results to the problem of orthogonal separation of the Hamilton-Jacobi equation. This work is motivated primarily by the case of spaces of constant curvature where warped products are abundant. We first characterize Killing tensors which have a natural algebraic decomposition in warped products. We then apply this result to show how one can obtain the Killing-St\"{a}ckel space (KS-space) for separable coordinate systems decomposable in warped products. This result in combination with Benenti's theory for constructing the KS-space of certain special separable coordinates can be used to obtain the KS-space for all orthogonal separable coordinates found by Kalnins and Miller in Riemannian spaces of constant curvature. Next we characterize when a natural Hamiltonian is separable in coordinates decomposable in a warped product by showing that the conditions originally given by Benenti can be reduced. Finally we use this characterization and concircular tensors (a special type of torsionless conformal Killing tensor) to develop a general algorithm to determine when a natural Hamiltonian is separable in a special class of separable coordinates which include all orthogonal separable coordinates in spaces of constant curvature.
\end{abstract}
\tableofcontents
\cleardoublepage
\phantomsection		

\printglossaries
\cleardoublepage
\phantomsection

\addcontentsline{toc}{section}{List of New Results}
\textbf{List of New Results}
\theoremlisttype{allname}
\listtheorems{thmMy,propMy,corMy}
\cleardoublepage
\phantomsection

\addcontentsline{toc}{section}{List of Known Results}
\textbf{List of Known Results}
\theoremlisttype{allname} 
\listtheorems{theorem,proposition,corollary}
\cleardoublepage
\phantomsection

\addcontentsline{toc}{section}{List of Examples}
\textbf{List of Examples}
\theoremlisttype{allname} 
\listtheorems{example}
\cleardoublepage
\phantomsection

\pagenumbering{arabic}
\newpage
\section{Introduction}



Eisenhart first showed that a special type of Killing $2$-tensor hereafter after called a characteristic Killing tensor, can be used to intrinsically characterize coordinates which orthogonally separate the geodesic Hamilton-Jacobi equation \cite{Eisenhart1934}. Benenti has generalized this result to Hamiltonians with potentials and to the non-orthogonal case \cite{Benenti1997a}.

Before continuing a discussion of the important results of the theory, we introduce some crucial notions. A \emph{separable web} is a collection $\Ei = (E_i)_{i=1}^{n}$ of $n$ pair-wise orthogonal $1$-distributions (line bundles) $E_{i}$, where each $E_{i}$ is spanned by one of the $n$ coordinate vector fields of a separable coordinate system. In other words, $(E_i)_{i=1}^{n}$ are the $n$ eigenspaces of a characteristic Killing tensor associated with a separable coordinate system. A \emph{warped product} is a product manifold $M = \prod_{i=0}^{k} M_{i}$ of pseudo-Riemannian manifolds $(M_{i}, g_{i})$ where $\dim M_{i} > 0$ for $i > 0$ equipped with the metric

\begin{equation}
	g = \pi_{0}^{*} g_{0} + \sum_{i=1}^{k} \rho_{i}^{2} \pi_{i}^{*} g_{i}
\end{equation}

\noindent where $\rho_{i} : M_0 \rightarrow \R^{+}$ are functions and $\pi_{i} : M \rightarrow M_{i}$ are the canonical projection maps \cite{Meumertzheim1999}. A separable web $\Ei = (E_i)_{i=1}^{n}$ is called \emph{reducible} if there exists a warped product $M = \prod_{i=0}^{k} M_{i}$ such that each $E_{i}$ is a section of the tangent bundle $T L_{j} \subset T M$ for some $j$ where $L_{j}$ is the canonical foliation associated with $M_{j}$. In order to make this last definition non-trivial, we must have $\dim M_{0} > 0$ or $k > 1$. Finally for the rest of the introduction and through most of the article, \emph{tensor} refers to a $2$-tensor.

Based on Eisenhart's initial work in \cite{Eisenhart1934}, Kalnins and Miller have obtained a recursive classification of all separable webs for Riemannian spaces of constant curvature \cite{Kalnins1986a,Kalnins1986}. An examination of their classification reveals that many of the separable webs they discovered are reducible. This was our initial motivation for studying Killing tensors and warped products.

It is known that to each separable web there are $n$ (including the metric) independent Lie-Schouten commuting Killing tensors diagonalized in any coordinates adapted to the separable web \cite{Benenti2004}. These Killing tensors span an $n$ dimensional vector space called the \defn{kss} associated to the separable web \cite{Benenti2004}. 

Working along the lines of the general theory of orthogonal separation in terms of characteristic Killing tensors, Benenti discovered a special class of separable webs for which the KS-space could be obtained algebraically \cite{Benenti1992c} using a special conformal Killing tensor hereafter called a Benenti tensor. These results are primarily applicable to spaces of constant curvature, for example to the elliptic and parabolic webs \cite{Benenti1992c}. They have been further refined in \cite{Benenti1993,Benenti2004,Benenti2005a}. \citeauthor{Crampin2003} showed \cite{Crampin2003} that a Benenti tensor $L$ could be intrinsically characterized as a symmetric contravariant tensor field whose associated $\binom{1}{1}$-tensor has point-wise simple real eigenvalues satisfying the following equation:

\begin{equation}
	\nabla_{x}L = \alpha \odot x
\end{equation}

\noindent for all vector fields $x$ and some vector field $\alpha$, where $\nabla$ is the Levi-Civita connection and $\gls{odot}$ denotes the symmetric product. He went on to study the remarkable properties possessed by these tensors (some of which we shall briefly review later). Hereafter we refer to a symmetric contravariant tensor $L$ satisfying the above equation as a \defn{ct}. We use this name because concircular tensors can be viewed as generalizations of concircular vectors \cite{Crampin2007}. Concircular tensors have also been called special conformal Killing tensors by \citeauthor{Crampin2003} in \cite{Crampin2003} and J-tensors by Benenti in \cite{Benenti2005a} where they are studied more thoroughly. One of our discoveries is that point-wise diagonalizable concircular tensors have a fundamental role to play in orthogonal separation of the Hamilton-Jacobi equation in spaces of constant curvature. This role will become clear after reading \cref{sec:BEKMsepAl}.

An unsolved problem within the theory is that of obtaining a basis for the KS-space for the separable webs found by Kalnins and Miller \cite{Kalnins1986} by using an algebraic procedure such as that in \cite{Benenti1992c}. The results obtained by Benenti in \cite{Benenti1992c} could be used to obtain a basis for the KS-space associated with certain ``irreducible'' separable webs obtained by Kalnins and Miller such as the elliptic and parabolic webs. In this article we define the notion of a \defn{kem} \emph{web}, which is (roughly speaking) a web built up recursively by using concircular tensors. We show that all the separable metrics found by Kalnins and Miller in \cite{Kalnins1986} are associated to KEM webs, hence showing that all separable webs in Riemannian spaces of constant curvature are KEM webs. In combination with other results presented in this article and by generalizing Benenti's theory presented in \cite{Benenti1992c}, one can show that a basis for the KS-space of any KEM web can be obtained algebraically using concircular tensors provided the warped product decompositions of the space are known. The specifics of this result will appear elsewhere. We note that warped product decompositions of spaces of constant curvature can be obtained in a straightforward manner and that they are well known in $\E^{n}$ and $\Si^{n}$ (some of which are closely related to spherical coordinates), see \cite{Nolker1996} for more details. In a subsequent expository article (see \cite{Rajaratnam2014c}) we shall describe the warped product decompositions of all spaces of constant curvature similar in style to \cite{Nolker1996}.

Another related problem we solve in this article is that of separating natural Hamiltonians defined in spaces of constant curvature. Taking advantage of the fact that separable webs in such spaces are KEM webs, we present a recursive algorithm (called the \defn{bekm} separation algorithm) which uses concircular tensors to determine when a natural Hamiltonian defined over a space of constant curvature is separable. This algorithm is equivalent to that derived by \citeauthor{Waksjo2003} in \cite{Waksjo2003} (by doing lengthy coordinate calculations based on St\"{a}ckel theory) for $\E^{n}$ and $\Si^{n}$. It should be noted that the BEKM separation algorithm is applicable to hyperbolic space and other signatures such as Minkowski space-time.

The approach taken in this article is based on the discovery (to be presented) that a multidimensional eigenspace of a concircular tensor\footnote{To be precise, we always assume these tensors (their associated endomorphisms) are point-wise diagonalizable.} naturally induces a warped product. Generalizing this fact, we make use of a formulation of Killings equation in terms of the geometry of the eigenspaces of a Killing tensor, originally given in \cite{Coll2006}. We use this formulation together with the theory of twisted and warped products presented in \cite{Meumertzheim1999} to show that a general (conformal) Killing tensor naturally induces a twisted product, see \cref{cor:CKTsandTPnets}. Continuing in this more general setting, in \cref{prop:extCKT} we characterize the Killing tensors which have a natural algebraic decomposition in a warped product. Building on this result, in \cref{prop:wpKss} we characterize the KS-space of a reducible separable web. Then we move on to study the problem of separating natural Hamiltonians in reducible separable webs; this culminates in \cref{thm:wpSOV}.

Finally in the last section, we present the main results of this article concerning the application of concircular tensors to the orthogonal separation of the Hamilton-Jacobi equation. This is done by applying the theory presented earlier which covered general warped products. In particular, we present some preliminary theory on concircular tensors, apply this theory to introduce the notion of KEM webs and then prove their relation to separable webs in spaces of constant curvature. We conclude by introducing the BEKM separation algorithm which can determine separability of natural Hamiltonians in KEM webs.

The results presented in this article and much more can be found in thesis of the first author \cite{Rajaratnam2014}.

\section{Preliminaries and Notation} \label{sec:preNot}

All differentiable structures are assumed to be smooth (class $C^{\infty}$). Without further specification the assumed context is a pseudo-Riemannian manifold $M$ of dimension $n$ equipped with covariant metric $g$. The contravariant metric is usually denoted by $G$ and $\bp{\cdot, \cdot}$ plays the role of the covariant and contravariant metric depending on the arguments. We denote $\gls*{spm}$ as the set of symmetric contravariant tensor fields of valence p on M. Furthermore $\defn{fm} = S^{0}(M)$ is the set of functions from M to $\R$ and $\defn{vem}$ denotes the set of vector fields over $M$. If $f \in \F(M)$ then $\nabla f \in \ve(M)$ denotes the gradient of $f$, i.e. the vector field metrically equivalent to $\d f$. Given a vector bundle E, $\defn{secd}$ denotes the set of sections of E.

We shall refer to a distribution $E$ as \emph{non-degenerate} if the induced metric on $E$ is non-degenerate at each point. The following notion of orthogonal nets will be useful:

\begin{definition}[Orthogonal Nets \cite{Meumertzheim1999}]
	A family $\Ei = (E_{i})_{i=1}^{k}$ of non-degenerate integrable distributions $E_{i}$ on a pseudo-Riemannian manifold $M$ is called an \emph{orthogonal net} on $M$ if the tangent bundle $T M$ can be decomposed as:
	
	\begin{equation}
		T M = \bigobot_{i=1}^{k}E_{i}
	\end{equation}
	
	Note that equations such as above are interpreted point-wise and that the symbol $\gls*{obot}$ stands for ``orthogonal direct sum''.
\end{definition}

Our remaining notations related to nets mostly follow the notations in \cite{Meumertzheim1999}. Suppose $M = \prod_{i=1}^{k} M_{i}$ is a pseudo-Riemannian product. We denote $M_{i \perp} := M_{1} \times \cdots \times M_{i-1}  \times M_{i+1} \times \cdots \times M_{k} $ and the canonical projections $\pi_{i} : M \rightarrow M_{i}$ by $p \rightarrow p_{i}$ for each $i$. We denote by $L_{i}$ the canonical foliation of $M$ induced by $M_{i}$. For $\bar{p} \in M$, the leaf of $L_{i}(\bar{p})$ through $\bar{p}$ and the canonical embedding of $M_{i}$ in M denoted $\tau_{i}$ are given by

\begin{align}
	\tau_{i}(p) & := (\bar{p}_{1},\dotsc, \bar{p}_{i-1}, p, \bar{p}_{i+1},\dotsc, \bar{p}_{k}),  \quad p \in M_{i} \\ 
	L_{i}(\bar{p}) &: = \tau_{i}(M_{i}) = \{p \in M :  p = \tau_{i}(p_{i}), \ p_{i} \in M_{i}  \}  
\end{align}

We let $E_{i}$ denote the integrable distribution induced by $L_{i}$, then $\Ei = (E_{i})_{i=1}^{k}$ is called the \emph{product net} of $\prod_{i=1}^{k} M_{i}$.

We can naturally ``lift'' any tensor defined on the manifolds $M_{i}$ to M. For example if $\tilde{\varphi} \in \F(M_{i})$ then the lift is $\varphi := \tilde{\varphi} \circ \pi_{i} \in \F(M)$, we denote the set of all such functions on $M$ of this form by $\hat{\F}(M_{i})$. For $\tilde{v} \in \ve(M_{i})$, the lift is the unique vector field $v \in \ve(M)$ such that $(\pi_{i})_{*} v = \tilde{v}$ and $(\pi_{i \perp})_{*} v = 0$. Analogously we denote the set of all such vector fields on $M$ of this form by $\hat{\ve}(M_{i})$. $\hat{S}^{p}(M_{i})$ is defined similarly. Note that it if $v \in \hat{\ve}(M_{i})$ and $u \in \hat{\ve}(M_{j})$, then $(\pi_{i})_{*} [v,u] = [\tilde{v}, \tilde{u}]$ if $i = j$ and $[v,u] = 0$ if $i \neq j$, where $[\cdot, \cdot]$ is the Lie bracket. Also note that usually we will use the same symbol for a tensor and its lift.  For $\varphi \in \F(M)$, we say that $\varphi$ is independent of $M_{i}$ (or $E_{i}$) if $\varphi \in \hat{\F}(M_{i \perp})$; if $M$ is connected this is equivalent to $\varphi_{*} E_{i} = 0$. We say that $\varphi$ depends only on $M_{i}$ (or $E_{i}$) if $\varphi \in \hat{\F}(M_{i})$.

A net $\Ei$ is said to be \emph{(locally) integrable} (or \emph{locally decomposable} in \cite{Meumertzheim1999}) if for every $p \in M$ there exists a neighborhood $U \subseteq M$ of $p$ and a $C^{\infty}$-diffeomorphism $f$ from a product manifold $\prod_{i=1}^{k} M_{i}$ onto $U$ such that for every $q \in \prod_{i=1}^{k} M_{i}$ and every $i = 1,...,k$ the slice $(q_{1},...,q_{i-1}) \times M_{i} \times (q_{i+1},...,q_{k})$ gets mapped into an integral manifold of $E_{i}$. In this case, the product manifold $\prod_{i=1}^{k} M_{i}$ is said to be (locally) \emph{adapted} to $\Ei$. A net $\Ei$ is called an \emph{(orthogonal) web} if it is integrable and $\dim E_{i} = 1$ for each i. Given a collection of distributions $\Ei = (E_{i})_{i=1}^{k}$ on a pseudo-Riemannian manifold, we say the collection is \emph{orthogonally integrable} if $\Ei$ forms an integrable net. In \cite[Theorem~1]{Reckziegel1999} the following has been shown, which justifies the term ``orthogonally integrable''

\begin{theorem}[Characterizations of integrable nets \cite{Reckziegel1999}] \label{thm:charIntNets}
	For the decomposition $T M = \bigobot_{i=1}^{k}E_{i}$ by the family of distributions $\Ei = (E_{i})_{i=1}^{k}$, the following are equivalent
	
	\begin{enumerate}
		\item $\Ei$ is an integrable net.
		\item The orthogonal distributions $E_{i}^{\perp}$ are integrable for $i=1,...,k$.
		\item The distributions $E_{i}$ and their direct sums $E_{i} \obot E_{j}$ are integrable for $i,j = 1,...,k$.
	\end{enumerate}
\end{theorem}

Without further specification, \emph{tensor} is short for valence 2-tensor field and being on a pseudo-Riemannian manifold, the type depends on the context. If T is an endomorphism of the tangent bundle, and $\lambda$ is an eigenfunction of T, then the eigenspace corresponding to $\lambda$ is $E_{\lambda} := \ker(T - \lambda I)$. A distribution $D$ is called \emph{T-invariant} if $T_{p} D_{p} \subseteq D_{p}$ for all $p \in M$. By an \emph{orthogonal tensor}, we mean a symmetric contravariant tensor whose uniquely determined endomorphism of $T M$ is point-wise diagonalizable with real eigenvalues. One can check that the eigenspaces of such an endomorphism (which is a \emph{self-adjoint} operator) are necessarily pair-wise orthogonal non-degenerate subspaces. An endomorphism T is said to have a \emph{simple eigenvalue} $\lambda$, if $\lambda$ is real and has algebraic multiplicity equal to 1. T is said to have \emph{simple eigenvalues} if all its eigenvalues are simple. We say an endomorphism field T has \emph{simple eigenfunctions} if it's eigenfunctions are point-wise simple. Finally note that all our arguments are local, we work in a neighborhood of a point where the dimensions of the eigenspaces don't vary, hence we can assume the eigenspaces are distributions.

\subsection{Brief outline of The Classification of Pseudo-Riemannian Distributions }

The following exposition of the classification of pseudo-Riemannian distributions is a combination of that from \cite{Meumertzheim1999} and \cite{Coll2006}. Suppose $E$ is an m-dimensional non-degenerate distribution. Then we use the orthogonal splitting $T M = E \obot E^{\perp}$, $V = V^{E} + V^{E^{\perp}}$, to define a tensor $s^{E} : T M \times_{M} E \rightarrow E^{\perp}$ and a linear connection $\nabla^{E}$ for $E$ by:

\begin{equation}
	\nabla_{X}Y = \nabla^{E}_{X}Y + s^{E}(X,Y)
\end{equation}

\noindent for all $X \in \ve(M)$ and $Y \in \Gamma(E)$. $s^{E}$ is called the \emph{generalized second fundamental form} of $E$ and above equation is referred to as the \emph{Gauss equation}. One can also check that $\nabla^{E}$ is metric compatible, i.e. $X\bp{Y,Z} = \bp{\nabla^{E}_{X}Y,Z} + \bp{Y,\nabla^{E}_{X}Z}$ for all  $X \in \ve(M)$ and $Y,Z \in \Gamma(E)$.

For the remainder of the discussion we set $s^{E} := s^{E}|_{(E\times_{M} E)}$. For $X,Y \in \Gamma(E)$, we can further decompose $s^{E}(X,Y)$ into its anti-symmetric part and symmetric part 

\begin{align}
	s^{E}(X,Y) & = (\nabla_{X}Y)^{E^{\perp}} = \frac{1}{2}(\nabla_{X}Y + \nabla_{Y}X)^{E^{\perp}} + \frac{1}{2}(\nabla_{X}Y - \nabla_{Y}X)^{E^{\perp}} \\
	& = h^{E}(X,Y) + A^{E}(X,Y)  \\
	A^{E}(X,Y) & := \frac{1}{2}(\nabla_{X}Y - \nabla_{Y}X)^{E^{\perp}} \\
	h^{E}(X,Y) & := \frac{1}{2}(\nabla_{X}Y + \nabla_{Y}X)^{E^{\perp}}
\end{align}

Since $\nabla$ is torsion-free, $A^{E}(X,Y) = \frac{1}{2}([X,Y])^{E^{\perp}}$, hence $E$ is integrable iff $A^{E} \equiv 0$. $h^{E}$ is called the \emph{second fundamental form} of $E$ and it is the second fundamental form of the leaves of the foliation induced by $E$ when $E$ is integrable \cite[P.~100]{barrett1983semi}. The second fundamental form can be decomposed in terms of its trace to get a further classification of $E$ as follows:

\begin{align}
	h^{E}(X,Y) & = \bp{X,Y} H_{E} + h^{E}_{T}(X,Y) \\
	 H_{E} & = \frac{1}{m} \tr{h^{E}}
\end{align}

\noindent where $h^{E}_{T}$ is trace-less. $H_{E}$ is called the \emph{mean curvature normal} of $E$. $E$ is called \emph{minimal}, \emph{umbilical} or \emph{geodesic}\footnote{Note that some authors use the name auto-parallel instead \cite{Meumertzheim1999}.} if $s^{E}(X,Y) = h^{E}_{T}(X,Y)$, $s^{E}(X,Y) = \bp{X,Y} H_{E}$ or $s^{E}(X,Y) = 0$ respectively for all $X,Y \in \Gamma(E)$. We add the qualifier ``almost'' to the three definitions above by replacing $s^{E}$ with $h^{E}$; this just drops the requirement that $A^{E} \equiv 0$. For example $E$ is \emph{almost umbilical} iff $h^{E}_{T} = 0$. We remark that when $E$ is one dimensional $h^{E}_{T} = 0$ trivially, hence all one dimensional non-degenerate foliations and similarly all one dimensional pseudo-Riemannian submanifolds are trivially umbilical. If E is umbilical and $\nabla_{X}^{E^{\perp}}H_{E} = 0$ for all $X \in \Gamma(E)$ then E is called \emph{spherical}. Finally if E is spherical and $E^{\perp}$ is geodesic then E is called \emph{Killing}; we will see later that Killing distributions can be thought of as multidimensional generalizations of Killing vectors.

\subsection{Twisted and Warped Products}

The content of this section is primarily from \cite{Meumertzheim1999} where the notion of a twisted product was introduced. For more on warped products see \cite{Meumertzheim1999,Zeghib2011}. The following general definition of a twisted product is useful in the study of conformal Killing tensors.

\begin{definition}[Twisted and Warped Products] \label{def:twistedProd}
	Let $M = \prod_{i=0}^{k} M_{i}$ be a product of pseudo-Riemannian manifolds $(M_{i}, g_{i})$ where $\dim M_{i} > 0$ for $i > 0$. Suppose for $i=0,...,k$, $\pi_{i} : M \rightarrow M_{i}$ is the projection map and $\rho_{i}: M \rightarrow \R^{+}$ is a function. The following metric $g$ on $M$ is called a \emph{twisted product metric}
	
	\begin{equation}
		g(X,Y) = \sum_{i=0}^{k} \rho_{i}^{2} g_{i}(\pi_{i *}X,\pi_{i *}Y) \quad \text{for $X,Y \in \ve(M)$}
	\end{equation}
	
	In this case $(M, g)$ is called a \emph{twisted product} and is denoted by $\sideset{^{\rho}}{_{i=0}^{k}}{\prod}M_{i}$ where $\rho = (\rho_{0},...,\rho_{k})$. Furthermore the $\rho_{i}$ are called twist functions of the twisted product. If each $\rho_{i}$ depends only on $M_{0}$ and $\rho_{0} \equiv 1$ then $g$ is called a \emph{warped product metric} and $(M,g)$ is called a \emph{warped product}. The warped product is denoted by $M_{0} \times_{\rho_{1}} M_{1} \times \cdots \times_{\rho_{k}} M_{k} $. $M_{0}$ is called the geodesic factor of the warped product and the $M_{i}$ for $i > 0$ are called spherical factors.
\end{definition}

\begin{example}
	By taking $M_{0}$ to be a point and $k = 1$ in the definition of a twisted product, we get a \emph{conformal product}.
\end{example}
\begin{example}
	By taking $M_{0}$ to be a point and $k > 1$ in the definition of a warped product, we get a \emph{pseudo-Riemannian product}. Throughout this article we will treat pseudo-Riemannian products as special cases of warped products this way.
\end{example}
\begin{example}
	If $\dim M_{i} = 1$ for each i, then the twisted product metric is locally the metric of an \emph{orthogonal coordinate system}.
\end{example}
\begin{example}[Prototypical warped product]
	The prototypical example of a warped product is the following warped product defined in (an open subset of) $\E^{n}$, which is the product manifold $\R^{+} \times S^{n-1}$ equipped with the metric $g = \d \rho^{2} + \rho^{2} \tilde{g}$ where $\tilde{g}$ is the metric of the $(n-1)$-sphere $S^{n-1}$.
\end{example}

Note that a twist function $\rho_{i}$ of a twisted product is only uniquely defined modulo products of functions $f \in \hat{\F}(M_{i})$. To elaborate, from the above definition one sees that we can multiply $\rho_{i}^{2}$ by $f \in \hat{\F}(M_{i})$ if we divide $g_{i}$ by $f$. The geometry of the twisted product is not altered by such transformations as we will see. We say that the twist functions are \emph{normalized} (with respect to a point $\bar{p} \in M$), if for each $i$, $\rho_{i}(p) = 1$ for all $p \in L_{i}(\bar{p})$. The following properties of the twisted product can be found in Proposition~2 in \cite{Meumertzheim1999}.

\begin{proposition}[Properties of the Twisted Product \cite{Meumertzheim1999}] \label{prop:tpProps}
	Let $\sideset{^{\rho}}{_{i=0}^{k}}{\prod}M_{i}$ be a twisted product with product net $\Ei = (E_{i})_{i=0}^{k}$ and $U_{i} := - \nabla \log \rho_{i}$.
	\begin{enumerate}
		\item $\Ei$ is an orthogonally integrable net.
		\item For each i the distribution $E_{i}$ is umbilical with mean curvature normal $H_{i} = U_{i}^{\perp i}$.
		\item $E_{i}$ is geodesic iff $\rho_{i}$ is independent of $M_{j}$ for $j \neq i$. $E_{i}^{\perp}$ is geodesic iff $\rho_{j}$ is independent of $M_{i}$ for $j \neq i$.
		\item If $\rho$ is independent of $M_{i}$ then $E_{i}$ is Killing. The converse is also true if the twisted product is normalized.
	\end{enumerate}
\end{proposition}

The following notions of twisted and warped product nets will be useful for studying conformal Killing tensors. It was originally Definition~3 in \cite{Meumertzheim1999}.

\begin{definition}[Twisted and warped product nets]
	Let M be a pseudo-Riemannian manifold and suppose $\Ei = (E_{i})_{i=0}^{k}$ is an orthogonal net.
	\begin{enumerate}
		\item $\Ei$ is called a \defn{tpnet} if it is integrable and each distribution $E_{i}$ is umbilical.
		\item $\Ei$ is called a \defn{wpnet} if $E_{i}$ is Killing for $i=1,...,k$.
	\end{enumerate}
\end{definition}
\begin{remark}
	In all applications, $\dim E_{i} > 0$ for $i > 0$. Although we will allow $\dim E_{0} = 0$ for a WP-net since this gives us a \emph{pseudo-Riemannian product net (RP-net)}.
\end{remark}

It can be shown that if $\Ei$ is a WP-net, then it is a TP-net with $E_{0} = \bigcap\limits_{i=1}^{k} E_{i}^{\perp}$ a geodesic distribution  \cite[Proposition~3]{Meumertzheim1999}. Also in the case $\Ei$ is a WP-net we refer to $E_{0}$ as the geodesic distribution of the WP-net and the $E_{i}$ for $i > 0$ as the Killing distributions of the WP-net. The following theorem gives the motivation for the above definition, it shows that every TP-net (resp. WP-net) admits a locally adapted twisted product (resp. warped product). See Corollary~1 in \cite{Meumertzheim1999} for a proof.

\begin{theorem}[Twisted and warped product nets \cite{Meumertzheim1999}] \label{cor:tpWpNets}
	Let M be a pseudo-Riemannian manifold and suppose $\Ei = (E_{i})_{i=0}^{k}$ is a TP-net (resp. WP-net). Then for every $p \in M$ there exists an open set $U \subseteq M$ containing p and a map $f : \prod_{i=0}^{k}M_{i} \rightarrow U $ which is an isometry with respect to a twisted (resp. warped) product metric on $\prod_{i=0}^{k}M_{i}$.
\end{theorem}
\begin{remark}
	One can also check that a similar theorem holds for a RP-net and a pseudo-Riemannian product metric.
\end{remark}

Now we can give some justification to the name ``Killing'' for a non-degenerate distribution which is spherical and has a geodesic orthogonal complement. By the above corollary, we see that a one dimensional Killing distribution is always spanned by a Killing vector field. Conversely any normal non-null Killing vector field spans a Killing distribution. The following can be said about multidimensional Killing distributions via the warped products they induce \cite{Zeghib2011}:

\begin{proposition}[Lifting isometries from Killing distributions]
	Let $M = B \times_{\rho} F$ be a warped product and suppose $\tilde{f} : F \rightarrow F$ is an isometry of F. Then the lift $f$ defined by
	
	\begin{equation}
		f(x,y) := (x, \tilde{f}(y)), \quad (x,y) \in B \times F
	\end{equation}
	
	\noindent is an isometry of M.
\end{proposition}

\subsection{Conformal Killing tensors}

A tensor $K \in S^{p}(M)$ is said to be a \defn{ckt} of valence p if there exists $C \in S^{p-1}(M)$ (called the \emph{conformal factor}) such that

\begin{equation}
	[K,G] = - 2 C \odot G
\end{equation}

\noindent where $[\cdot,\cdot]$ is the \emph{Schouten bracket}\footnote{Note that the Schouten bracket for symmetric tensors is directly related to the Poisson bracket on the cotangent bundle \cite{Woodhouse1975a,Nijenhuis1955}.} \cite{Woodhouse1975a,Nijenhuis1955} which generalizes the usual Lie bracket of vector fields. When $C = 0$, $K$ is called a \defn{kt} of valence p and additionally if $p = 1$ then $K$ is a \defn{kv} and the above equation reduces to $\lied{G}{K} = 0$. In terms of the Levi-Civita connection $\nabla$, for a Killing tensor, the above equation becomes in coordinates:

\begin{equation}
	\nabla_{(i}K_{i_{1}...i_{p})} = 0
\end{equation}

\noindent which implies that $K_{i_1 \ldots i_p}\dot{x}^{i_1}\ldots\dot{x}^{i_p}$ is constant along geodesics where $(x^{i}(t))$ are parametrized geodesics in local coordinates.

We now enumerate special classes of conformal Killing tensors that are of interest. If $K \in S^{2}(M)$ is a Killing tensor, we say it is a \defn{chkt} if it has simple eigenfunctions and its eigenspaces are orthogonally integrable. Due to \cref{thm:charIntNets}, the condition that the eigenspaces are orthogonally integrable is equivalent to the condition that $K$ has normal eigenvector fields\footnote{A normal vector field is a non-zero vector field whose orthogonal distribution is Frobenius integrable.}. If $K$ is a CKT with conformal factor $C = \nabla f$ for some $f \in \F(M)$, we say $K$ is a CKT of \emph{gradient-type}. In particular if $f = \tr{K}$, then we say $K$ is of \emph{trace-type}.

\subsection{Orthogonal Separation of The Hamilton-Jacobi Equation}

Suppose M is a pseudo-Riemannian manifold and denote by $T^{*}M$ the cotangent bundle of M. Suppose $(q,p)$ are canonical coordinates on $T^{*}M$ and let $V \in \F(M)$. Then the (natural) \emph{Hamiltonian} $H$ is defined by:

\begin{equation}
	H(q,p) := \frac{1}{2} \bp{p,p} + V(q)
\end{equation}

The \emph{geodesic Hamiltonian} is obtained by setting $V \equiv 0$ in the above equation.

Separable coordinates $(q^{i})$ on $M$ are characterized as solutions of the Levi-Civita equations \cite{Levi-Civita1904} (see \cite[P.~13]{Kalnins1986} for English readers). By analyzing these equations, one can show that a given natural Hamiltonian is separable in given coordinates only if the geodesic Hamiltonian is separable in those coordinates \cite{Benenti1997a}. One can also show that if a given orthogonal coordinate system,$(q^{i})$, is separable, then any coordinates adapted to the orthogonal web formed by $(q^{i})$ are also separable\footnote{We should mention that our definition of an orthogonal web is dual to the one used in \cite{Benenti1997a}, which defines an orthogonal web as $n$ pair-wise orthogonal co-dimension one non-degenerate foliations.}. Hence orthogonal separation depends only on the existence of a special orthogonal web, hereafter called an (orthogonal) \emph{separable web}. By further analysis of the Levi-Civita equations, one can characterize orthogonal separation of geodesic Hamiltonians in terms of ChKTs:

\begin{theorem}[Orthogonal Separation of Geodesic Hamiltonians \cite{Benenti1997a}] \label{thm:HJosepI}
	The geodesic Hamiltonian is separable in an orthogonal web $\Ei$ iff there exists a ChKT whose eigenspaces 
	form $\Ei$.
\end{theorem}

Hence the problem of classifying separable coordinates for a geodesic Hamiltonian is equivalent to the problem of classifying ChKTs. Before presenting results on the separation of natural Hamiltonians, we need a further definition. By analyzing the equations satisfied by a ChKT in adapted coordinates, one can show that there is an $n$-dimensional vector space of KTs simultaneously diagonalized in the adapted coordinates. This vector space of KTs is called the \defn{kss} associated with an orthogonal separable web.

The following theorem addresses the separation of natural Hamiltonians:

\begin{theorem}[Benenti's Theorem \cite{Benenti1997a}] \label{thm:HJosepII}
	A natural Hamiltonian with potential $V$ is separable in a web $\Ei$ iff there exists a ChKT $K$ whose eigenspaces form $\Ei$ which satisfies the \emph{dKdV equation}:

	\begin{equation} \label{eq:dKdV}
		\d (K \d V) = 0
	\end{equation}
	
	Furthermore if $V$ separates in the separable web $\Ei$, then all $K$ in the KS-space associated with $\Ei$ satisfy the dKdV equation with $V$.
\end{theorem}

\section{Orthogonal conformal Killing tensors}

For now, by a (conformal) Killing tensor we will mean an orthogonal (conformal) Killing tensor. In this section we fist present a formulation of the (conformal) Killing equation in terms of the eigenspaces of a (conformal) Killing tensor given in \cite{Coll2006}. This formulation will be the most useful in our study. We then use the theory of twisted and warped products given in \cite{Meumertzheim1999} to show that an orthogonal (conformal) Killing tensor naturally induces a twisted product and then derive the well known (conformal) Killing equation in the eigenframe. We then give necessary and sufficient conditions on an eigenfunction of the tensor for the associated eigenspace to be geodesic or Killing. Finally we end with a well known result on restricting CKTs to special submanifolds which will be used later.

The following theorem can be deduced from Theorem~2 in \cite{Coll2006} which only covered traceless orthogonal CKTs. Before we state it, for $x,y \in \ve(M)$ we let $\{x,y\} := \frac{1}{2}(\nabla_{x}y + \nabla_{y}x)$. Also, given a collection of distributions $(E_{i})_{i=1}^{k}$ satisfying $T M = \bigobot_{i=1}^{k}E_{i}$, then for any vector $x \in \ve(M)$, we have the orthogonal splitting $x = \sum\limits_i x^i$ where each $x^i \in \Gamma(E_{i})$.

\begin{theorem}[Geometric Characterization of Orthogonal CKTs \cite{Coll2006}] \label{prop:OKTcharII}
	Let T be an orthogonal tensor and let $E_{i}$ be the eigenspaces corresponding to the eigenfunctions $\lambda_{i}$. Then T is a conformal Killing tensor with conformal factor t iff
	
	\begin{enumerate}
		\item The eigenspaces $E_{i}$ are almost umbilical.
		\item The mean curvature normal $H_i$ of the eigenspace $E_i$ satisfies the following equation:
		
		\begin{equation}
			H_{i}  = - \frac{1}{2} \sum\limits_{j \neq i} (\nabla \log \Abs{\lambda_{i}-\lambda_{j}})^{j}
		\end{equation}
		\item The conformal factor satisfies the following equation:
		\begin{equation}
			t = \sum (\nabla \lambda_{i})^{i}
		\end{equation}
		\item $T(\{x,y\},z) + T(\{z,x\},y) + T(\{y,z\},x) = 0$ for eigenvectors $x,y,z$ with different eigenfunctions
	\end{enumerate}
\end{theorem}
\begin{remark}
	For a Killing tensor the second condition can be simplified to:
	
	\begin{align}
		H_{i}  & = - \frac{1}{2} \sum\limits_{j \neq i} \frac{1}{(\lambda_{i}-\lambda_{j})} (\nabla \lambda_{i})^{j}
	\end{align}
\end{remark}

We will now proceed to show that when the eigenspaces are orthogonally integrable, Condition 4 of the above theorem is automatically satisfied. The following lemma can be deduced from a knowledge of rotation coefficients, although we state it for completeness.
\begin{lemma} \label{lem:ldncov}
	Suppose $(E_{i})_{i=0}^{k}$ is an integrable net. Then for $x \in \Gamma(E_{i})$ and $y \in \Gamma(E_{j})$ with $j \neq i$, $\nabla_{x}y \in \Gamma(E_{i} \obot E_{j})$.
\end{lemma}
\begin{proof}
	Suppose $z \in \Gamma(E_{k})$ where k is different from $i,j$. Observe that
		
	\begin{equation}
		g(\nabla_{x}y, z) - g(\nabla_{y}x, z) = g([x,y], z) = 0
	\end{equation}
	
	Also
	
	\begin{equation}
		g(\nabla_{y}x, z) + g(x, \nabla_{y}z) = \nabla_{y}g(x,z) = 0
	\end{equation}
	
	\noindent The above two equations hold for all permutations of $x,y,z$. Thus
	
	\begin{multline}
		g(\nabla_{x}y, z) = g(\nabla_{y}x, z) = - g(x, \nabla_{y}z) = -  g(x, \nabla_{z}y)  \\
		= g(\nabla_{z} x, y) = g(\nabla_{x} z, y) = - g( z, \nabla_{x}y)
	\end{multline}
	
	Thus $g(\nabla_{x}y, z) = 0$.
\end{proof}

The following corollary gives a version of the above theorem for orthogonal tensors with orthogonally integrable eigenspaces.
\begin{corMy} \label{prop:OKTcharIII}
	Suppose T is an orthogonal tensor with orthogonally integrable eigenspaces and let $E_{i}$ be the eigenspaces corresponding to the eigenfunctions $\lambda_{i}$.
	Then T is a conformal Killing tensor with conformal factor t iff
	
	\begin{enumerate}
		\item The eigenspaces $E_{i}$ are umbilical.
		\item The mean curvature normals of the eigenspaces satisfy the following equation:
		
		\begin{equation} \label{eq:CKTmeanCurv}
			H_{i}  = - \frac{1}{2} \sum\limits_{j \neq i} (\nabla \log \Abs{\lambda_{i}-\lambda_{j}})^{j}
		\end{equation}
		\item The conformal factor satisfies the following equation:
		\begin{equation} \label{eq:CKTconFac}
			t = \sum (\nabla \lambda_{i})^{i}
		\end{equation}
	\end{enumerate}
\end{corMy}
\begin{proof}
	Since $A^{E_{i}} = 0$ for each i, the eigenspaces are almost umbilical iff they are umbilical. Condition 4 of \cref{prop:OKTcharII} is automatically satisfied due to \cref{lem:ldncov}, hence the result holds by \cref{prop:OKTcharII}.
\end{proof}

Now we use a result from \cite{Meumertzheim1999} which characterizes twisted products to show that orthogonally integrable CKTs naturally give rise to a twisted product structure.

\begin{corMy}[Conformal Killing tensors induce twisted product nets] \label{cor:CKTsandTPnets}
	Suppose T is an orthogonal tensor with orthogonally integrable eigenspaces $(E_{i})_{i=1}^{k}$ and associated eigenfunctions $(\lambda_{i})_{i=1}^{k}$. Let $M = \prod\limits_{i=1}^{k} M_{i}$ be a connected product manifold locally adapted to the eigenspaces of T. Then T is a CKT iff $(M,g)$ is a twisted product with twist functions $\rho_{i}$ satisfying the following equation:
	
	\begin{equation} \label{eq:meanCurvWarpFn}
		(\nabla \log \rho_{i})^{\perp i} = \frac{1}{2} \sum\limits_{j \neq i} (\nabla \log \Abs{\lambda_{i}-\lambda_{j}})^{j}
	\end{equation}
\end{corMy}
\begin{proof}
	This result follows from the above corollary together with \cref{cor:tpWpNets} and \cref{prop:tpProps}~(2).
\end{proof}

The above corollary contains new information mainly when at least one of the eigenspaces has dimension greater than 1. Indeed when the eigenfunctions of T in the above corollary are simple and T is a KT, then Eisenhart has solved the defining equations in \cite{Eisenhart1934}. He has shown that the metric is in St\"{a}ckel form and has given the relationship of the KT to the metric via the St\"{a}ckel matrix \cite{Eisenhart1934}. Although knowing that all separable metrics are in St\"{a}ckel form is of little use for finding the ChKTs defined on a given pseudo-Riemannian manifold. The progress made by \citeauthor{Eisenhart1934} in \cite{Eisenhart1934} and then \citeauthor{Kalnins1986a} in \cite{Kalnins1986a} on obtaining orthogonal separable coordinates in $S^{n}$ and $\E^{n}$ are based on the integrability conditions of \cref{eq:meanCurvWarpFn,eq:CKTconFac} when $T$ is a ChKT.

The above corollary motivates us to define a \emph{Killing net} (K-net) and a \emph{Conformal Killing net} (CK-net) as the TP-net formed by the eigenspaces of a Killing tensor respectively Conformal Killing tensor when the eigenspaces are orthogonally integrable. The following proposition shows that CK-nets are a special class of TP-nets. In particular, it will give us a simple way to check when an eigenspace of a CKT is Killing.

\begin{propMy} \label{prop:perpGeoImpKil}
	Suppose $(E_{i})_{i=1}^{k}$ is an orthogonally integrable CK-net and let $\lambda_{i}$ be the associated eigenfunctions. If $E_{i}^{\perp}$ is geodesic, then $E_{i}$ is spherical.
\end{propMy}
\begin{proof}	
	Suppose $x \in \hat{\ve}(M_{i})$ and $y \in \hat{\ve}(M_{j})$ where $j \neq i$. Then by \cref{eq:CKTmeanCurv}
	
	\begin{align}
		x \bp{H_{i}, y} & = - \frac{1}{2} x \bp{\nabla \log \Abs{\lambda_{i}-\lambda_{j}},y} \\
		& = - \frac{1}{2} x y  \log \Abs{\lambda_{i}-\lambda_{j}} \\
		& = - \frac{1}{2} y x  \log \Abs{\lambda_{i}-\lambda_{j}} \\
		& = - \frac{1}{2} y \bp{\nabla \log \Abs{\lambda_{i}-\lambda_{j}},x} \\
		& = y \bp{H_{j},x}
	\end{align}
	Now, since $E_{i}^{\perp}$ is geodesic, one can show that $H_{j}^{i} = 0$ for $j \neq i$. This can be seen for example, by working in a local twisted product given by \cref{cor:CKTsandTPnets} and then using \cref{prop:tpProps}~(3). Hence by the above calculation, $x \bp{H_{i}, y} = y \bp{H_{j}, x} = y \bp{H_{j}^{i}, x} = 0$. Thus
	
	\begin{align}
		\bp{\nabla_{x} H_{i}, y} & = x \bp{H_{i}, y} - \bp{H_{i}, \nabla_{x} y} \\
		& = - \bp{H_{i}, \nabla_{y} x} \\
		& = \bp{\nabla_{y} H_{i}, x}  - y \bp{H_{i}, x} \\
		& = \bp{\nabla_{y} H_{i}, x} \\
		& = 0
	\end{align}
	
	\noindent where the last line follows since $E_{i}^{\perp}$ is geodesic. Hence $\bp{\nabla_{x} H_{i}, y} = 0$ for all $x \in \Gamma(E_{i})$ and $y \in \Gamma(E_{i}^{\perp})$, thus $E_{i}$ is spherical.
\end{proof}

The following corollary allows us to determine the geometry of the eigenspaces of a CKT with orthogonally integrable eigenspaces using its eigenfunctions.

\begin{corMy} \label{cor:CKTeigProps}
	Suppose T is a CKT with conformal factor $t$ and orthogonally integrable eigenspaces $(E_{i})_{i=1}^{k}$. 
	\begin{enumerate}
		\item $E_{i}$ is Killing iff
			
			\begin{equation}
				(\nabla \lambda_{j})^{i}  = t^{i} \quad \text{for all j } \neq i 
			\end{equation}
		\item $E_{i}$ is geodesic iff
		
			\begin{equation}
				(\nabla \lambda_{i})^{j} = t^{j}  \quad \text{for all j } \neq i 
			\end{equation}
	\end{enumerate}
	
	In particular for a KT, $E_{i}$ is Killing iff all the eigenfunctions are independent of $E_{i}$ and $E_{i}$ is geodesic iff $\lambda_{i}$ is a constant.
\end{corMy}
\begin{proof}
	This follows from the above proposition together with \cref{prop:OKTcharIII} and the definitions of Killing and geodesic distributions.
\end{proof}

From the above corollary, it follows immediately that if $M$ admits a KT with orthogonally integrable eigenspaces $\Ei =(E_{i})_{i=0}^{k}$ and respective eigenfunctions $(\lambda_{i})_{i=0}^{k}$ such that $\lambda_{0}$ is constant and $\lambda_{i}$ depends only on $E_{0}$ for each $i > 0$, then $\Ei$ is a WP-net. One can easily use \cref{prop:OKTcharIII} and \cref{cor:CKTsandTPnets} to show conversely that any WP-net admits a KT. Although in the next section this fact will follow as a corollary of another proposition we will prove.

If $D$ is a distribution then we denote by $S^{p}(D)$ the set of symmetric contravariant tensors of valence $p$ over the vector bundle $D$. The following proposition on restriction of CKTs to submanifolds will be of use later on.

\begin{proposition}[Restriction of CKTs to Invariant Submanifolds] \label{prop:restrictCKT}
	Let $T$ be a CKT with conformal factor $t$ and suppose $D$ is an integrable non-degenerate T-invariant distribution. If $\tilde{M}$ is an integral manifold of $D$ regarded as a pseudo-Riemannian manifold with the induced metric, then T restricts to a CKT on $\tilde{M}$ with the induced conformal factor.
\end{proposition}
\begin{proof}
	By hypothesis $TM = D \obot D^{\perp}$, hence we can write
	
	\begin{align}
		T & = T_{D} + T_{D^{\perp}} \\
		t & = t_{D} + t_{D^{\perp}} \\
		G & = G_{D} + G_{D^{\perp}}
	\end{align}
	
	Let $\iota : \tilde{M} \rightarrow M$ be the inclusion map, then note that $T_{D} = \iota_{*} \tilde{T}$ for some $\tilde{T} \in S^{2}(\tilde{M})$. Similar equations hold for $t_{D}$ and $G_{D}$. Thus we observe that the following equation holds over $\tilde{M}$, $[T_{D}, G_{D}] = [\iota_{*} \tilde{T}, \iota_{*} \tilde{G}] = \iota_{*} [ \tilde{T}, \tilde{G}]$ by naturality of the Schouten bracket. In particular, we see that $[T_{D}, G_{D}] \in S^{2}(D)$. Now
	
	\begin{align}
		[T,G] & = [T_{D}, G_{D}] + [T_{D}, G_{D^{\perp}}] + [T_{D^{\perp}}, G_{D}] + [T_{D^{\perp}}, G_{D^{\perp}}]
	\end{align}
	
	also
	
	\begin{align}
		t \odot G & = t_{D} \odot G_{D} + t_{D} \odot G_{D^{\perp}} + t_{D^{\perp}} \odot G_{D} + t_{D^{\perp}} \odot G_{D^{\perp}}
	\end{align}
	
	By projecting onto $S^{2}(D)$ we find that $[T_{D}, G_{D}] = -2 t_{D} \odot G_{D}$, thus $[ \tilde{T}, \tilde{G}] = -2 \tilde{t} \odot \tilde{G}$ by injectivity of $\iota_{*}$.
\end{proof}

\begin{remark}
	The above result shows that any ChKT $K$ induces a ChKT on any leaf of the foliation of a K-invariant distribution. Hence by \cref{thm:HJosepI}, this gives a method to construct lower dimensional separable webs from a given separable web. Conversely it motivates us to look for methods to build separable webs of higher dimension from given separable webs. We will see later on that warped products give us a means to do just this.
\end{remark}

\section{Killing tensors in Warped Products} \label{sec:KTinWP}

In this section we give conditions under which $K \in S^{2}(M)$ that admits a K-invariant Killing distribution is a \gls{kt}. The first application of this result is to find necessary and sufficient conditions for extending Killing tensors defined on the geodesic and spherical factors of a warped product. This result also allows us to study \glspl{kss} which have a common invariant Killing distribution. We will see that such KS-spaces can be decoupled into ones on the geodesic and spherical factors of an induced warped product. Such decomposable KS-spaces will allow us to define a \emph{reducible separable web} which is a natural concept that follows from this investigation.

\subsection{Preliminaries}

But first we need some properties of the Schouten bracket.

\begin{lemma}
	Suppose M is a manifold with local coordinates $(x^{i})$ and let $X_{i} := \partial_{i}$. Then for $K,G \in S^{2}(M)$ the Schouten bracket is given as follows

	\begin{align}
		[K,G] & = (2 K^{i l} (\d G^{j k})_{l} - 2 G^{i l} (\d K^{j k})_{l}) X_{i} \odot X_{j} \odot X_{k} \label{eq:SchoutForm} \\
		& = 2 ( K \d G^{j k} -  G \d K^{j k}) \odot X_{j} \odot X_{k}
	\end{align}
	
	Furthermore, the following hold:
	
	For $V \in \ve(M)$ and $K \in S^p (M)$
	\begin{equation}
		[V, K] = \lied{K}{V}
	\end{equation}
	
	For $G \in S^{2}(M)$ and $f \in \F(M)$
	\begin{equation}
		[G,f] = 2 G(\d f)
	\end{equation}
\end{lemma}
\begin{proof}
	For the proof of the first statement, see for example \cite{Woodhouse1975a}. The second statement is well known, see also for example \cite{Woodhouse1975a}. The third can also be deduced from \cite{Woodhouse1975a}, but we give the proof here. First note that $G^{ij} X_{i}\otimes X_{j} = G^{ij} X_{i}\odot X_{j}$, hence
	
	\begin{align}
		[G,f] & = [G^{ij}X_{i}\odot X_{j}, f] \\
		& = G^{ij}X_{i} \odot [ X_{j}, f] + [G^{ij}X_{i}, f] \odot X_{j} \\
		& = G^{ij}X_{i} \odot [ X_{j}, f] + G^{ij} [X_{i}, f] \odot X_{j} \\
		& = 2 G^{ij}[ X_{i}, f]  X_{j} \\
		& = 2 G(\d f)
	\end{align}
\end{proof}

The following lemma won't be directly used but it's useful to keep it in mind for proofs to come.

\begin{lemma}[Schouten bracket on Product Manifolds]
	Let $M = B \times F$ be a product manifold and suppose $K \in \hat{S}^{p}(B)$, $G \in \hat{S}^{q}(F)$. Then the following holds:
	
	\begin{equation}
		[K,G] = 0 
	\end{equation}
\end{lemma}
\begin{proof}
	This follows from the naturality of the Schouten bracket, i.e. the proof is similar to that when $K$ and $G$ are vector fields.
\end{proof}

\subsection{Killing tensors in Warped Products}

In the follow proposition we will characterize KTs in warped products.

\begin{propMy}[Killing tensors in Warped Products] \label{prop:extCKT}
	Suppose $K \in S^{2}(M)$ and $D$ is a K-invariant Killing distribution. Let $B \times_{\rho} F$ be a local warped product adapted to the WP-net $(D^{\perp},D)$ with contravariant metric $G  = G_{0} + \kappa G_{1}$ where $\kappa := \rho^{-2}$.
	 
	 Then $K$ is a KT iff there exist KTs $K' \in S^{2}(B)$, $\tilde{K} \in S^{2}(F)$ and $t \in \F(B)$ such that the following equations hold:
	 
	 \begin{align}
	 	K & = K' + t G_{1} +  \tilde{K} \label{eq:WPktDecomp} \\
	 	\d t & = K' \d \kappa
	 \end{align}
	 
	 Furthermore $\tilde{K}$ is also a KT on $B \times_{\rho} F$.
\end{propMy}
\begin{proof}
	By hypothesis, we can write $K = K_{0} + K_{1}$ where $K_{0}  \in S^{2}(D^{\perp})$ and $K_{1}  \in S^{2}(D)$. Thus,
	
	\begin{align} 
		[K,G] & = [K_{0} + K_{1} ,G_{0} + \kappa G_{1} ] \\ 
		& = [K_{0} ,G_{0} ] + [K_{0}  , \kappa G_{1} ] + [ K_{1} ,G_{0}] + [ K_{1} ,\kappa G_{1} ] \\ 
		& = [K_{0} ,G_{0} ] + \kappa [K_{0}  , G_{1} ] + 2 K_{0}(\d \kappa) \odot  G_{1}  + [ K_{1} ,G_{0}] + \kappa [ K_{1} , G_{1} ]
	\end{align}
	
	Note that $[K_{0} ,G_{0} ] \in S^{3}(D^{\perp})$ (see \cref{eq:SchoutForm}), then by linear independence, $[K,G] = 0$ iff
	
	\begin{align}
		[K_{0} ,G_{0} ] & = 0 \label{eq:wpKt1} \\
		[ K_{1} , G_{1} ] &  = 0 \label{eq:wpKt2} \\
		\kappa [K_{0}  , G_{1} ] + 2 K_{0}(\d \kappa) \odot  G_{1}  + [ K_{1} ,G_{0}] & = 0 \label{eq:wpKt3}
	\end{align}
	
	Suppose $(x^{i}) = (x^{a},x^{\alpha})$ are local coordinates adapted to the warped product $B \times_{\rho} F$. We denote coordinates for $B$ using Latin letters such as $a,b$, coordinates for $F$ using Greek letters such as $\alpha, \beta$ and the letters $i,j,k$ are reserved for generic indices.   Let $X_{i} := \partial_{i}$, then
	
	\begin{align}
		[ K_{1} ,G_{0}] & = 2 ( K_{1} \d G_{0}^{j k} -  G_{0} \d K_{1}^{j k}) \odot X_{j} \odot X_{k} \\
		& = - 2 G_{0} \d K_{1}^{\alpha \beta} \odot X_{\alpha} \odot X_{\beta} 
	\end{align}
	
	and
	
	\begin{align}
			[K_{0}  , G_{1}] & = 2 ( K_{0} \d G_{1}^{j k} -  G_{1} \d K_{0}^{j k}) \odot X_{j} \odot X_{k} \\
			& = - 2 G_{1} \d K_{0}^{a b} \odot X_{a} \odot X_{b} 
	\end{align}
	
	Thus by linear independence, \cref{eq:wpKt3} is satisfied iff
	
	\begin{align}
		[K_{0}  , G_{1}] & = 0 \\
		2 K_{0}(\d \kappa) \odot  G_{1}  + [ K_{1} ,G_{0}] & = 0
	\end{align}
	
	The first of the above equations are satisfied iff $G_{1} \d K_{0}^{a b} = 0$, i.e. $K_{0} \in \hat{S}^{2}(B)$. The second becomes 
	
	\begin{align}
		2 K_{0}(\d \kappa) \odot  G_{1}  + [ K_{1} ,G_{0}] & = 2 (K_{0}(\d \kappa) G_{1}^{\alpha \beta} - G_{0} \d K_{1}^{\alpha \beta}) \odot X_{\alpha} \odot X_{\beta} 
	\end{align}
	
	\noindent which is identically zero iff
	
	\begin{align}
		\d^{0} K_{1}^{\alpha \beta} & = K_{0}(\d \kappa) G_{1}^{\alpha \beta} \\
		\Rightarrow \d^{0} (K_{0}(\d \kappa)) & = 0 \quad \text{by non-degeneracy of $G_{1}$}
	\end{align}
	
	\noindent where $\d^{0}$ is $\d$ followed by the (point-wise) orthogonal projection onto $(D^\perp )^*$. So, $K_{0}(\d \kappa) = \d^{0} t$ for some $t \in \F(B)$, thus
	
	\begin{align}
		\d^{0} K_{1}^{\alpha \beta} & = \d^{0} (t G_{1}^{\alpha \beta})
	\end{align}
	
	Hence $\tilde{K}^{\alpha \beta} := K_{1}^{\alpha \beta} - t G_{1}^{\alpha \beta} \in \F(F)$, i.e. $\tilde{K} \in \hat{S}^{2}(F)$. \Cref{eq:wpKt3} is satisfied iff $\tilde{K}$ is a KT on $F$ and \cref{eq:wpKt1} is satisfied iff $K_{0}$ is a KT on $B$. Finally if we let $K' := K_{0}$, the result follows. The last statement that $\tilde{K}$ is a KT on $B \times_{\rho} F$ can be readily verified from the above equations.
\end{proof}

Two important special cases of the above proposition are the following:
\begin{enumerate}
	\item By taking $K', t =0$, we see that $\tilde{K} \in \hat{S}^{2}(F)$ is a KT on F iff it is a KT on $B \times_{\rho} F$.
	\item By taking $\tilde{K} = 0$ we find that a necessary and sufficient condition for $K' \in S^{2}(B)$ to be lifted into a KT on $B \times_{\rho} F$ is that
	\begin{equation}
		\d (K' \d \kappa) = 0
	\end{equation}
\end{enumerate}

The second case reveals a connection between extending KTs into warped products and the separation of the Hamilton-Jacobi equation for natural Hamiltonians. Indeed, suppose $V \in \F(M)$ is the potential function of a natural Hamiltonian and let $K \in S^{2}(M)$ be a ChKT. Now consider the local warped product $M \times_{\rho} \E^{1}_{\nu}$ where $\rho,\nu$ are defined as follows:

\begin{align}
	\frac{1}{\rho^{2}} & := 2 \Abs{V} & \nu & := \sgn{V}
\end{align}

\noindent in a neighborhood of a point where $V$ is non-zero. This warped product metric is called an \emph{Eisenhart metric}, since Eisenhart showed that geodesics $x^{i}(t)$ in this warped product with $\dot{x}^{n+1} = 1$ project onto solutions of Hamilton's Equations for the natural Hamiltonian associated with $V$ \cite{Eisenhart1928}. Now observe that the condition for extending $K$ into a KT on the warped product $M \times_{\rho} \E^{1}_{\nu}$ given above is precisely the condition that determines weather or not $V$ is separable in the orthogonal web associated with $K$ by \cref{thm:HJosepII}. This connection was observed by Benenti \cite{Benenti1997a} when he proved \cref{thm:HJosepII}. We will use this connection later on to derive necessary and sufficient conditions for extending a KS-space from the geodesic factor of a warped product.

We can also prove the following corollary cf. \cite{Jelonek2000}, which shows that a WP-net is a K-net.

\begin{corollary}[WP-nets always admit KTs] \label{cor:wpnetKT}
	A pseudo-Riemannian manifold M admits a WP-net $\Ei = (E_{i})_{i=0}^{k}$ iff there exists a KT, $K$ on M whose eigen-net is $\Ei$ and the corresponding eigenfunctions $\lambda_{i}$ satisfy:
	\begin{enumerate}
		\item $\lambda_{0}$ is a constant
		\item $\lambda_{i}$ depends only on $E_{0}$ for each $i > 0$
	\end{enumerate}
	
	Furthermore if such a KT exists, then the warping functions can locally be chosen to satisfy the following equation $\rho_{i}^{2} = \Abs{\lambda_{i}-\lambda_{0}}$ for $i > 0$.
\end{corollary}
\begin{proof}
	If $M$ admits a KT with orthogonally integrable eigenspaces and eigenfunctions satisfying the above conditions, then it follows from \cref{cor:CKTeigProps} that its eigenspaces form a WP-net.
	
	Conversely suppose $\Ei$ is a WP-net, and suppose $ G = G_{0} + \sum_{i=1}^{k} \kappa_{i} G_{i}$ is an adapted warped product metric. The above proposition shows that each $G_{i}$ for $i > 0$ is a KT on $M$. Hence for each $i$ if we choose $c_{i} \in \R$, then $K := c_{0} G + \sum_{i=1}^{k} c_{i} G_{i}$ is a KT on $M$. Thus, locally we can always choose the $c_{i}$ such that $K$ is a KT with eigenspaces equal to $\Ei$ and clearly the eigenfunctions satisfy the above conditions.
	
	Now if such a KT exists, by \cref{eq:meanCurvWarpFn} in \cref{cor:CKTsandTPnets}, we have for $i > 0$
	
	\begin{align}
		(\nabla \log \rho_{i}^{2})^{\perp i} & = \sum\limits_{j \neq i} (\nabla \log \Abs{\lambda_{i}-\lambda_{j}})^{j} \\
		& = \nabla \log \Abs{\lambda_{i}-\lambda_{0}}
	\end{align}
	
	Thus it follows that locally we can choose the warping functions as stated.
\end{proof}

The following corollary follows immediately by inductively applying \cref{prop:extCKT}.

\begin{corMy} \label{cor:wpGenKTdecomp}
	Suppose $\Ei  = (D_{i})_{i=0}^{k}$ is a WP-net and $K$ is KT with $D_{i}$ a K-invariant distribution for $i = 1,\dotsc,k$. Let $M = M_{0} \times_{\rho} \prod_{i=1}^{k} M_{i}$ be a local warped product adapted to $\Ei$. Then in contravariant form, K can be decomposed as follows:
	
	\begin{equation}
		K = K_{0} + \sum_{i=1}^{k} K_{i}
	\end{equation}
	
	\noindent where each $K_{i} \in \hat{S}^{2}(M_{i})$ is a KT for $i = 1,..,k$. Furthermore $K_{0}$ is a KT and each $D_{i}$ is an eigenspace of $K_{0}$ for $i = 1,..,k$ cf. \cref{cor:CKTeigProps}.
\end{corMy}

There are many examples of warped product metrics in relativity where these results could be applied. We will give two such examples in this article, here is the first:

\begin{example}[Separation in FRW metrics]
	An FRW metric is the product manifold $M = \E_{1}^{1} \times_{\rho} F$ equipped with warped product metric $g = - \d t^2 + \rho^2\tilde{g}$ where $\tilde{g}$ is a Riemannian metric (assumed to be of constant curvature which is not necessary here).
	
	The distribution orthogonal to $\pderiv{}{t}$, is Killing. Then by \cref{prop:extCKT}, note that the metric on $\E_{1}^{1}$ can be lifted to a KT on $M$. Hence it follows by \cref{prop:extCKT} (at least locally) that $M$ admits a ChKT $K$ with timelike eigenvector field $\pderiv{}{t}$ iff there exists a ChKT $\tilde{K} \in S^2(F)$. So any ChKT $\tilde{K} \in S^2(F)$ together with the time coordinate $t$ induces a separable web on $M$.
\end{example}

\subsection{Killing-St\"{a}ckel spaces in Warped Products}

The following corollary uses the connection between extending KTs into warped products and the separation of potentials observed earlier to obtain a necessary and sufficient condition for extending a Killing-St\"{a}ckel space from the geodesic factor of a Warped Product.

\begin{corMy}[Extending a Killing-St\"{a}ckel space into a Warped Product] \label{cor:WPextKS}
	Suppose $M = B \times_{\rho} F$ is a warped product and $\K$ is Killing-St\"{a}ckel space in $B$. If there exists a ChKT $K \in \K$ that can be extended into a KT on $M$ (via the method of \cref{prop:extCKT}) then all KTs in $\K$ can be extended into KTs on $M$.
\end{corMy}
\begin{proof}
	Suppose $K \in \K$ is a ChKT that can be extended into a KT on $M$. Then from \cref{prop:extCKT}, $K$ satisfies the dKdV equation with $\rho^{-2}$. Then by Benenti's theorem it follows that every $K \in \K$ satisfies the dKdV equation with $\rho^{-2}$, hence by the above proposition every $K \in \K$ can be extended into a KT on $M$.
\end{proof}

The above corollary motivates the following notion of a reducible separable web, which is characterized intrinsically by the invariant distributions of an associated ChKT.

\begin{definition}[Reducible separable web]
	Suppose $\Ei$ is an orthogonal separable web locally characterized by a ChKT, $K$. $\Ei$ is said to be reducible if it admits a K-invariant Killing distribution.
\end{definition}

First note that since all KTs in the KS-space of a separable web are simultaneously diagonalized, the above definition doesn't depend on the choice of the ChKT, hence is well-defined. One can check that the above definition of a reducible separable web is equivalent to the one given in the introduction. We will make exclusive use of the above definition in the rest of the article. The following proposition states clearly why we introduce to notion of reducible separable webs.

\begin{propMy}[The Killing-St\"{a}ckel space of a reducible separable web] \label{prop:wpKss}
	Suppose $K$ is a ChKT with associated KS-space $\K$ inducing a reducible separable web, i.e. there exists a K-invariant Killing distribution $D$. Let $M = B \times_{\rho} F$ be a local warped product adapted to the WP-net $(D^{\perp},D)$ with adapted contravariant metric $G = G_{B} + \rho^{-2} G_{F}$. Then there are KS-spaces $\K_{B}$ and $\K_{F}$ on $B$ and $F$ respectively such that $L \in \K$ iff there exists $L_{B} \in \K_{B}$, $L_{F} \in \K_{F}$ and $l \in \hat{\F}(B)$ such that the following equations hold
	 \begin{align}
	 	L & = L_{B} + l G_{F} +  L_{F} \\
	 	\d l & = L_{B} \d \rho^{-2}
	 \end{align}
\end{propMy}
\begin{proof}
	By \cref{prop:restrictCKT} it follows that $\K$ induces a KS-space $\K_{B}$ in $B$ and a KS-space $\K_{F}$ in $F$. If $L \in \K$, then it follows from \cref{prop:extCKT} that $L$ is determined up to constants by KTs in $\K_{B}$ and $\K_{F}$ satisfying the above equations.  Conversely from \cref{prop:extCKT} it follows that every KT in $\K_{F}$ can be extended to a KT in $\K$. Furthermore it follows from \cref{prop:extCKT} that $K$ can be decomposed into a KT on $M$ to satisfy the hypothesis of the above corollary. Hence from the above corollary it follows that each $L_{B} \in \K_{B}$ can be extended into a KT in $\K$ given by the above equation by taking $L_{F} = 0$.
\end{proof}

One usually determines if an orthogonal separable web is reducible by inspecting the metric in adapted coordinates by using \cref{prop:tpProps}~(4) and keeping in mind that all KTs in the KS-space are diagonalized in adapted coordinates. We give some examples to illustrate this.

\begin{example}
	The dimension of the Killing distribution is one in the above definition iff there is a Killing vector spanning one of the distributions of the web. This is sometimes called a \emph{web symmetry} \cite{Horwood2009}.
\end{example}
\begin{example}
	There is an abundant supply of reducible separable webs in spaces of constant curvature \cite{Kalnins1986}. These are a special case of KEM webs which will be introduced in \cref{sec:KEMweb}.
\end{example}

\section{Separation of The Hamilton-Jacobi Equation in Warped Products} \label{sec:sepHJeqnWP}

In this section we are concerned with the separation of the Hamilton-Jacobi equation in reducible separable webs. $K$ is assumed to be a \gls{kt} with orthogonally integrable eigenspaces $(E_{i})_{i=1}^{k}$ with associated eigenfunctions $\lambda_{1},...,\lambda_{k}$ and multiplicities $m_{1},...,m_{k}$. We work in the local twisted product $\sideset{^{\rho}}{_{i=1}^{k}}{\prod}M_{i}$ adapted to the eigenspaces of K given by \cref{cor:CKTsandTPnets}.

Fix $x \in \hat{\ve}(M_{i})$ and $y \in \hat{\ve}(M_{j})$ such that $[x,y] = 0$, then letting $\sigma_{i} := \log \rho_{i}^{2}$, it follows from \cref{eq:meanCurvWarpFn} that the eigenfunctions satisfy

\begin{equation}
	x \lambda_{j} = (\lambda_{j} - \lambda_{i}) x \sigma_{j}
\end{equation}

\noindent and we note that the above equation holds even if $i = j$. Fixing $V \in \F(M)$ and using the above equation we have

\begin{align}
	\d (K \d V)(x,y) & = x(K(y, \nabla V)) - y(K(x, \nabla V)) \\
	& = x(\lambda_{j} y V) - y(\lambda_{i} x V) \\
	& = x \lambda_{j} y V - y \lambda_{i} x V + \lambda_{j} x y V - \lambda_{i} y x V \\
	& = (\lambda_{j} - \lambda_{i}) x \sigma_{j} y V - (\lambda_{i} - \lambda_{j}) y \sigma_{i} x V +  (\lambda_{j} - \lambda_{i})x y V \\
	& = (\lambda_{j} - \lambda_{i})(x y V + x \sigma_{j} y V + y \sigma_{i} x V )
\end{align}


Hence we have proven the following:

\begin{proposition}[The dKdV Equation in the eigenframe] \label{prop:solnsOfDKdVeQ}
	Given K and V as above, $\d (K \d V) \equiv 0$ iff for each $x \in \hat{\ve}(M_{i})$ and $y \in \hat{\ve}(M_{j})$ with $i \neq j$ the following holds:
	
	\begin{equation} \label{eq:dKdVconsq}
		x y V + x \log \rho_{j}^{2} y V + y \log \rho_{i}^{2} x V = 0
	\end{equation}
	
	From which we can deduce the following:
	
	\begin{enumerate}
		
		\item If $E_{i}^{\perp}$ is geodesic, hence $E_{i}$ is Killing (see \cref{{prop:perpGeoImpKil}}), we have for all $y \in \hat{\ve}(M_{i \perp})$:
		\begin{equation} \label{eq:KiCons}
			y( \rho_{i}^{2} x V) = 0
		\end{equation}
		
		\item In particular, if $E_{i}$ and $E_{j}$ are Killing and $i \neq j$, we have for $x \in \hat{\ve}(M_{i})$ and $y \in \hat{\ve}(M_{j})$:
		\begin{equation}
			x y V = 0
		\end{equation}
	\end{enumerate}
\end{proposition}
\begin{proof}
	The first equation immediately follows from the above calculations. Now for the consequences, if $E_{i}^{\perp}$ is geodesic, then $ x(\log \rho_{j}^{2}) = 0$ for $j \neq i$ by \cref{prop:tpProps}~(3), hence
	
	\begin{align}
		x y V + x \log \rho_{j}^{2} y V + y \log \rho_{i}^{2} x V & = x y V +  y \log \rho_{i}^{2} x V  \\
		& = x y V + \frac{y \rho_{i}^{2}}{\rho_{i}^{2}} x V \\
		& = \frac{1}{\rho_{i}^{2}} (\rho_{i}^{2} x y V + y \rho_{i}^{2} x V) \\
		& = \frac{1}{\rho_{i}^{2}}  y(\rho_{i}^{2} x V)
	\end{align}
	
	Hence $y(\rho_{i}^{2} x(V)) = 0$, the second statement also follows immediately.
\end{proof}


Now, suppose $E_{i}$ is Killing and that $\tilde{K}_{i}$ is a KT on $M_{i}$, then by \cref{prop:extCKT}, the lift $K_{i}$, is a KT on M. The following proposition will allow us to reduce the calculation of the dKdV equation with $K_{i}$ on M to the restriction of the equation on $M_{i}$. To make this precise, we fix $\bar{p} \in M$ and let $L_{i}(\bar{p})$ be the leaf of the canonical foliation of $M_{i}$ through $\bar{p}$. Furthermore let $\tau_{i} : M_{i} \rightarrow L_{i}(\bar{p})$ be the embedding of $M_{i}$ in $M$.

\begin{propMy}[Reduction of The dKdV equation on Warped Products] \label{prop:wpReddKdV}
	Suppose $K$ and $K_{i}$ are as above, $E_{i}$ is Killing and additionally assume that $M$ is connected. For a potential $V \in \F(M)$, let $V_{i} := \tau_{i}^{*} V \in \F(M_{i})$. Suppose $\d{(K \d{V})} = 0$ holds on $M$, then the following is true:
	
	\begin{align}
		\d{(K_{i} \d{V})} = 0 \quad \Leftrightarrow \quad \d{(\tilde{K}_{i} \d{V_{i}})} = 0
	\end{align}
\end{propMy}
\begin{proof}
	The first implication follows trivially by naturality of the exterior derivative, so now we prove the converse. First we note that as endomorphisms of $T^{*} M$, $K_{i} = \rho_{i}^{2} \tilde{K}_{i}$ where $\tilde{K}_{i}$ is the lift of an endomorphism of $T^{*} M_{i}$. We also note that for $y \in \hat{\ve}(M_{i \perp})$
	
	\begin{equation}
		\lied{ (\rho_{i}^{2} (\d V)_{i})}{y} = 0
	\end{equation}
	
	\noindent where $(\d V)_{i}$ is the orthogonal projection of $\d V$ onto $T^{*} M_{i}$.  To prove this, we first note that since $\d{(K \d{V})} = 0$, $y (\rho_{i}^{2} x V) = 0$ for all $x \in \hat{\ve}(M_{i})$ by \cref{eq:KiCons} in \cref{prop:solnsOfDKdVeQ}. This implies that $\d{(\rho_{i}^{2} (\d V)_{i})} = 0$. Hence the above equation follows by Cartan's Formula which relates the exterior derivative of forms to their Lie derivatives.
	
	Now by hypothesis, for $x \in \hat{\ve}(M_{i})$ and $y \in \hat{\ve}(M_{i})$ with $[x,y] = 0$ we have that $\tau_{i}^{*} (\d{(K_{i} \d{V})}(x,y)) = 0$. Then for $z \in \hat{\ve}(M_{i \perp})$,
	
	\begin{align}
		z \d{(K_{i} \d{V})}(x,y) & = z [x(K_{i}(y, \d V)) - y(K_{i}(x, \d V))] \\
		& =  z[x(\tilde{K}_{i}(y, \rho_{i}^{2} (\d V)_{i}))) - y(\tilde{K}_{i}(x, \rho_{i}^{2} (\d V)_{i})))] \\
		& =  x(\tilde{K}_{i}(y, \lied{ (\rho_{i}^{2} (\d V)_{i})}{z}))) - y(\tilde{K}_{i}(x, \lied{ (\rho_{i}^{2} (\d V)_{i})}{z}))) \\
		& = 0
	\end{align}
	
	\noindent where the last equation follows since $\lied{ (\rho_{i}^{2} (\d V)_{i})}{z} = 0$. Thus since M is connected we conclude that $\d{(K_{i} \d{V})}(x,y) = 0$ on M.
	
	For $x \in \hat{\ve}(M_{i})$ and $y \in \hat{\ve}(M_{i \perp})$
	
	\begin{align}
		\d{(K_{i} \d{V})}(x,y) & = x(\tilde{K}_{i}(y, \rho_{i}^{2} (\d V)_{i})) - y(\tilde{K}_{i}(x, \rho_{i}^{2} (\d V)_{i})) \\
		& = - \tilde{K}_{i}(x, y(\rho_{i}^{2} (\d V)_{i})) \\
		& = 0
	\end{align}
	
	Also it easily follows that for $x \in \hat{\ve}(M_{i \perp})$ and $y \in \hat{\ve}(M_{i \perp})$, that $\d{(K_{i} \d{V})}(x,y) = 0$. Thus the result is proven.
\end{proof}

We now consider the problem of separation in warped products. To be precise, suppose $N = N_{0} \times_{\rho} \prod\limits_{i=1}^{l} N_{i}$ is a warped product and $\Ei = (D_{i})_{i=0}^{l}$ is the associated WP-net. Suppose K is a ChKT such that each Killing distribution defining $\Ei$ is K-invariant. According to Benenti's Theorem (\cref{thm:HJosepII}), for a potential $V \in \F(M)$ to be separable in the web associated with K, we need to check that the dKdV equation is satisfied. Although in this case we have some more information. Due to \cref{cor:wpGenKTdecomp}, K can be decomposed as follows in contravariant form:

\begin{equation}
	K = K_{0} + \sum_{i=1}^{l} K_{i}
\end{equation}

\noindent where each $K_{i} \in \hat{S}^{2}(N_{i})$ is a KT for $i = 1,..,l$, each $D_{i}$ is an eigenspace of $K_{0}$ for $i = 1,..,l$ and $K_{0}$ restricted to $D_{0}$ is characteristic. By Benenti's Theorem, if V satisfies the dKdV equation with K, then it must satisfy the dKdV equation with each $K_{i}$. In particular it must satisfy the dKdV equation with $K_{0}$. Since $K_{0}$ invariantly encodes the warped product through its eigenspaces and a partial separable web on $D_{0}$, one could ask if the converse holds. If V satisfies the dKdV equation with a given KT $K_{0}$ with eigenspaces as just stated, is it possible to build up a separable web for V by reducing the problem to one on the spherical factors of $N$? The following theorem shows that we can.

\begin{thmMy}[Separation in Warped Products] \label{thm:wpSOV}
	Suppose $(D_{i})_{i=0}^{l}$ is a WP-net and $K_{0}$ is a KT with eigenspaces $D_{i}$ for $i = 1,...,l$ and characteristic on $D_{0}$. Fix $\bar{p} \in M$ and let $N = \prod\limits_{i=0}^{l} N_{i}$ be a connected product manifold passing through $\bar{p}$ adapted to the WP-net $(D_{i})_{i=0}^{l}$. Then the following holds:
	
	Suppose $V \in \F(M)$ satisfies $\d (K_{0} \d V) = 0$. Let $V_{i} := \tau_{i}^{*} V \in \F(N_{i})$ and suppose for each $i = 1,...,k$ there exists a ChKT $\tilde{K}_{i}$ on $N_{i}$ such that $\d (\tilde{K}_{i} \d V_{i}) = 0$.
	
	Then V is separable in the web formed by the simple eigenspaces of $K_{0}$ together with the lifts of the simple eigenspaces of $\tilde{K}_{1},...,\tilde{K}_{l}$.
\end{thmMy}
\begin{proof}
	For $i = 1,...,l$, let $K_{i}$ be the lift of $\tilde{K}_{i}$ to $N$. Consider the tensor

	\begin{equation}
		K := K_0 + \sum_{i=1}^{l} K_i
	\end{equation}

	By \cref{prop:extCKT}, $K$ is a Killing tensor on $N$. Let $\tilde{G}_i$ be the contravariant metric on $N_i$, then by replacing $\tilde{K}_{i}$ with $a_i \tilde{K}_{i} + b_i \tilde{G}_i$ for some $a_i \in \R \setminus \{0\}$ and $b_i \in \R$, we can assume $K$ locally has simple eigenfunctions. Let $q_0$ be coordinates which diagonalize the ChKT induced by $K_0$ on $N_0$. Let $q_j$ be coordinates which diagonalize $\tilde{K}_{j}$ on $N_j$ for each $j > 0$. Then one can check that the product coordinates $(q_0, q_1, \dotsc, q_l)$ are orthogonal and diagonalize $K$, hence $K$ is a ChKT. By \cref{prop:wpReddKdV}, $\d (K_{i} \d V) = 0$ on $N$ for each $i > 0$, hence $K$ satisfies the dKdV equation with $V$. Thus it follows by \cref{thm:HJosepII} that $V$ separates in the product coordinates $(q_0, q_1, \dotsc, q_l)$, which proves the claim.
\end{proof}

The above theorem and the preceding discussion shows that reducible separable webs enable one to reduce the problem of separation to certain spherical submanifolds after one finds a KT with the same eigenspaces as $K_{0}$ in the above theorem. In the following section, we shall see that in special cases, orthogonal concircular tensors generate such a KT. In fact in the following section, we will use orthogonal concircular tensors and the above theorem to develop a general algorithm for separating a potential.

\section{Main Application: Concircular tensors and The BEKM Separation Algorithm} \label{sec:BEKMsepAl}

In this section we apply the theory developed in this article to concircular tensors. We first review the theory of concircular tensors and then present the key observation that a multidimensional eigenspace of concircular tensor is necessarily Killing. Then we introduce the notion of KEM webs and prove their relation to separable webs in spaces of constant curvature. We conclude by introducing the BEKM separation algorithm which can determine separability of natural Hamiltonians in KEM webs.

\subsection{Concircular tensors}

Recall that $L \in S^{2}(M)$ is called a \gls{ct} if it satisfies the following equation

\begin{equation} \label{eq:CTform}
	\nabla_{x}L = \alpha \odot x
\end{equation}

\noindent for all $x \in \ve(M)$ and some vector field $\alpha$.

An \defn{oct} or more succinctly an OC-tensor is a concircular tensor which is also an orthogonal tensor. OC-tensors with simple eigenfunctions were studied extensively by Benenti, see \cite{Benenti1992c,Benenti2004,Benenti2005a}; thus in recognition of his contributions we refer to this special class of OC-tensors as \emph{Benenti tensors} (also called L-tensors by Benenti).

OC-tensors have some useful properties. First, given a tensor $L$, let $N_{L}$ be the Nijenhuis tensor (torsion) of $L$ \cite{Gerdjikov2008a}. We say that $L$ is \emph{torsionless} if its Nijenhuis tensor vanishes. Then if $L$ is a concircular tensor, the following equations hold \cite[Lemma 3.1]{Benenti2005a} (cf. \cite{Crampin2003})

\begin{align}
	[L, G] = - 2 \nabla \tr{L} \odot G \\
	N_{L} = 0 
\end{align}

Conversely, by Theorem 19.3 in \cite{Benenti2005a}, an orthogonal tensor satisfying the above equations is a C-tensor. The first of the above equations tells us that a C-tensor is a conformal Killing tensor of trace-type. The second equation can be interpreted if we assume $L$ is an OC-tensor.

Suppose $L$ is an OC-tensor with eigenspaces $(E_{i})_{i=1}^{k}$ and corresponding eigenfunctions $\lambda^{1},...,\lambda^{k}$. Since an OC-tensor has Nijenhuis torsion zero, by Theorem~13.29 in \cite{Gerdjikov2008a}, the eigenspaces $(E_{i})_{i=1}^{k}$  are orthogonally integrable and each eigenfunction $\lambda^{i}$ depends only on $E_{i}$, i.e.

\begin{equation} \label{eq:torProp}
	(\nabla \lambda_{i})^{j} = 0 \quad j \neq i
\end{equation}

Now we present the property of OCTs that connects them with warped products. The trace-type condition implies that the conformal factor $\alpha = \nabla \tr{L} = \sum\limits_{i} m_{i} \nabla \lambda_{i}$ where $m_{i} = \dim E_{i}$. This together with \cref{eq:CKTconFac} in \cref{prop:OKTcharIII} gives the following:
	
\begin{equation}
	m_{i} \nabla \lambda_{i} = \alpha^{i} = \nabla \lambda_{i}
\end{equation}

Hence when $\dim E_{i} > 1$, $\lambda_{i}$ must be constant; this property was first observed by Benenti in \cite[Theorem~A.5.1]{Benenti2005a}. Then in combination with \cref{eq:torProp}, \cref{cor:CKTeigProps} implies that $E_{i}$ is a Killing distribution.

In order to separate the Hamilton-Jacobi equation, we need Killing tensors. Thus we observe that the following tensor, called the \defn{kbdt} generated by $L$, is a KT:

\begin{equation} \label{eq:KBDT}
	K = \tr{L} G - L
\end{equation}

\noindent This follows by a direct calculation. An important observation is that $K$ has the same eigenspaces as $L$. We also note that if $L$ is a Benenti tensor, then its KBDT is a ChKT, hence by \cref{thm:HJosepI} it follows that the orthogonal web associated with $L$ is separable. Later on we will use the KBDT and \cref{thm:wpSOV} to present an algorithm for separating natural Hamiltonians.

But first, we need the following characterization of orthogonal CTs to make statements about the completeness of the algorithm to be presented.

\begin{theorem}[Characterization of orthogonal CTs] \label{prop:charOCT}
	Suppose $L$ is a torsionless orthogonal tensor with eigenspaces $(E_{i})_{i=1}^{k}$ and associated eigenfunctions $(\lambda_{i})_{i=1}^{k}$. Then L is an OCT iff there is a twisted product adapted to its eigenspaces such that each twist function $\rho_{i}$ can be chosen to be:
	\begin{equation} \label{eq:OCTwarpFn}
		\rho_{i}^{2} = \prod\limits_{k \neq i} \Abs{\lambda_{i}-\lambda_{k}}
	\end{equation}
	\noindent and each multidimensional eigenspace $E_{i}$ is a Killing distribution, or equivalently the eigenfunction corresponding to $E_{i}$ is constant.
\end{theorem}
\begin{proof}
	Since an OCT is equivalent to a trace-type torsionless orthogonal CKT \cite[Theorem~A.5.3]{Benenti2005a}, we only need to characterize a trace-type torsionless orthogonal CKT.

	By \cref{cor:CKTsandTPnets} there is a twisted product $\sideset{^{\rho}}{_{i=1}^{k}}{\prod}M_{i}$ which is adapted to the eigenspaces of L. We can explicitly solve for the twist function $\rho_{i}$ in this case. From \cref{eq:CKTmeanCurv}, we have
	
	\begin{align}
		H_{i} & = - \frac{1}{2} \sum\limits_{j \neq i} (\nabla \log \Abs{\lambda_{i}-\lambda_{j}})^{j} \\
		& = - \frac{1}{2} \sum\limits_{j \neq i} ( \sum\limits_{k \neq i} \nabla \log \Abs{\lambda_{i}-\lambda_{k}})^{j} \\
		& = - \frac{1}{2} \sum\limits_{j \neq i} ( \nabla \log \prod\limits_{k \neq i} \Abs{\lambda_{i}-\lambda_{k}})^{j} \\
		& = - \frac{1}{2} ( \nabla \log \prod\limits_{k \neq i} \Abs{\lambda_{i}-\lambda_{k}})^{\perp i}
	\end{align}
	
	Hence by \cref{eq:meanCurvWarpFn}, we have
	
	\begin{align}
		(\nabla \log \rho_{i}^{2})^{\perp i} = ( \nabla \log \prod\limits_{k \neq i} \Abs{\lambda_{i}-\lambda_{k}})^{\perp i}
	\end{align}
	
	Thus $\log \rho_{i}^{2} - \log \prod\limits_{k \neq i} \Abs{\lambda_{i}-\lambda_{k}} = f_{i}$ where $f_{i}$ is a function of $M_{i}$ only. Thus we see that a necessary and sufficient condition for $L$ to be a torsionless orthogonal CKT is that the twist functions have the form given by \cref{eq:OCTwarpFn}.
	
	We also note from earlier calculations that the only constraint imposed by the trace-type condition is that the multidimensional eigenspaces have constant eigenfunctions. Thus the conclusion follows by \cref{cor:CKTsandTPnets}.
\end{proof}
\begin{remark}
	This characterization (with some-what less information) originally appeared in \cite{Benenti2005a}.
\end{remark}

For illustrative purposes it will be useful to have the general concircular tensor in pseudo-Euclidean space $\eunn$, it is given as follows in contravariant form \cite[theorem~B.1.1]{Benenti2005a}:

\begin{equation} \label{eq:ctGen}
	L = A + m r \odot r + w \odot r
\end{equation}

\noindent where $A$ is symmetric and constant, $m \in \R$, $w \in \R^{n}$ and $r = (x^{1},\dotsc,x^{n})$ is the vector representing the generic point in $\eunn$ hereafter called the \emph{dilatational vector field}.

We will also need the following fact concerning C-tensors, it shows that they form a vector space and it gives an upper bound on the dimension of this space.

\begin{theorem}[The Vector Space of Concircular Tensors \cite{Thompson2005}] \label{thm:SCKTsFiniteDim}
	The C-tensors form a finite dimensional real vector space when $n > 1$ with maximal dimension equal to $\frac{1}{2}(n+1)(n+2)$. Furthermore the maximal dimension is achieved if and only if the space has constant curvature.
\end{theorem}
\begin{remark}
	See \cite{Crampin2007} and references therein for more on the possible dimensions of the vector space of C-tensors.
\end{remark}

\subsection{KEM webs and Separable Webs in Spaces of Constant Curvature} \label{sec:KEMweb}

In this section we first introduce a certain type of orthogonal web called a KEM web which follows naturally from orthogonal concircular tensors. We will show that KEM webs are necessarily separable. Then we will use the classification of separable webs in Riemannian spaces of constant curvature given in \cite{Kalnins1986} to show that all separable webs in these spaces are KEM webs.

Before we introduce the general notion of a KEM web, we first present the following simple motivating example:

\begin{example}[KEM webs] \label{ex:kemWeb}
	In this example we work in $\E^{3}$ with the CT $L = d \odot d$ where $d \neq 0$ is a constant vector. In this case, $L$ has a simple eigenspace $S_{1} := \spa{d}$ and a multidimensional eigenspace $D_{1} := d^{\perp}$. Clearly a warped product manifold adapted to the WP-net $(S_{1}, D_{1})$ is $\E^{1} \times \E^{2}$. 

	Now in $\E^{2}$ we can specify a Cartesian coordinate system via the CT $L = A$ where $A$ is symmetric, constant and has simple eigenspaces. We can also specify polar coordinates via the CT $L = r \odot r$ where $r$ is the dilatational vector field as in the previous section. In both cases it is well known that this defines a separable web $\Ei_{1}$ in $\E^{2}$.

	Back in $\E^{3}$ we can define an orthogonal web, $\Ei$, formed by $S_{1}$ together with the lift of $\Ei_{1}$ (which is obtained by translating $\Ei_{1}$ along $d$). In the first case we obtain a web defining Cartesian coordinates and in the second case we obtain a web defining cylindrical coordinates, both of which are separable.
\end{example}

We have shown two examples where an orthogonal (in fact separable) web was obtained recursively using concircular tensors. For low dimensions we define a KEM web as follows: When $n = 1$ the tangent bundle $T M$ itself is trivially defined to be a KEM web. When $n = 2$ any non-trivial\footnote{By a non-trivial concircular tensor, we mean one which is not a multiple of the metric when $n > 1$.} OCT has simple eigenfunctions, hence is a Benenti tensor and defines an orthogonal web. So when $n = 2$ we define a KEM web to be any orthogonal web associated with a Benenti tensor. In the general case we define recursively a KEM web as follows:

\begin{definition}[KEM web]
	Let $L$ be a non-trivial OCT with simple eigenspaces $(S_{i})_{i=1}^{k}$ and multidimensional eigenspaces $(D_{i})_{i=1}^{l}$. For each $i= 1,...,l$, let $\Ei_{i}$ be a KEM web on an integral manifold of $D_{i}$. Then the web formed by $(S_{i})_{i=1}^{k}$ together with the lifts of $\Ei_{i}$ is called a \emph{Kalnins-Eisenhart-Miller (KEM) web}.
\end{definition}
\begin{remark}
	One can check that the above definition is well-defined since each $D_{i}$ is necessarily integrable and the lift of $\Ei_{i}$ is necessarily an orthogonal web at least locally.
\end{remark}

\begin{propMy}[KEM webs] \label{prop:KEMisSep}
	A KEM web is a separable web.
\end{propMy}
\begin{proof}
	Suppose inductively that this theorem holds for all KEM webs with dimension $k < n$ and note that the statement trivially holds for $k = 1$ since the metric is always an OCT. Now we prove the proposition for KEM webs of dimension $n > k \geq 1$.

	Let $L$ be the OCT in the definition of the KEM web and let $K$ be the KBDT associated with $L$. Let $D_{1},...,D_{l}$ be the multidimensional eigenspaces of $L$, these are necessarily Killing distributions by \cref{prop:charOCT}. Then the net formed by $D_{1},...,D_{l}$ together with $D_{0} := \bigcap\limits_{i=1}^{l} D_{i}^{\perp}$ is a WP-net. So fix $\bar{p} \in M$ and let $N = \prod_{i=0}^{l} N_{i}$ be a connected product manifold adapted to this net and passing through $\bar{p}$. For each $i = 1,...,l$, let $K_{i}$ be a ChKT for $\Ei_{i}$ on $N_{i}$ which is given by \cref{thm:HJosepI}. It follows from \cref{prop:extCKT} that $K_{i}$ can be extended to a KT on $M$ (which we call $K_{i}$). After adding a constant multiple of the induced metric on $N_{i}$ to $K_{i}$ if necessary, we can assume that $K + \sum_{i=1}^{l} K_{i}$ is a ChKT at least locally. Since $K + \sum_{i=1}^{l} K_{i}$ is a ChKT for this KEM web, it follows from \cref{thm:HJosepI} that this KEM web is a separable web.
	
	Thus the result follows by induction.
\end{proof}

\begin{thmMy}[Separable Webs in Spaces of Constant Curvature] \label{thm:SCCconChKT}
	In a space of constant curvature, every separable web is a KEM web.
\end{thmMy}

To prove the above theorem, we need some preliminary results. The following lemma is well known, see for example \cite{Nolker1996}.

\begin{lemma} \label{lem:SCCspherFol}
	In a space of constant curvature, a Killing foliation is a foliation of homothetic\footnote{By homothetic pseudo-Riemannian manifolds, we mean conformal pseudo-Riemannian manifolds where the conformal factor is a positive constant.} spaces of constant curvature.
\end{lemma}


In particular for $\E^n$ one can show that a Killing foliation is foliation by subsets of (affine) spheres or planes of lesser dimension. The following theorem is the key to proving the above theorem.

\begin{thmMy}[KEM Separation Theorem] \label{thm:KEMsep}
	Suppose $K$ is a ChKT defined on a space of constant curvature $M$. Then there is a non-trivial concircular tensor $L$ defined on $M$ such that each eigenspace of $K$ is L-invariant, i.e. $L$ is diagonalized in coordinates adapted to the eigenspaces of $K$.
\end{thmMy}

A rigorous proof will be given in \cite{Rajaratnam2014d}. In Riemannian spaces of constant curvature, this theorem can be proven by connecting the classification of separable metrics given by Kalnins and Miller in \cite{Kalnins1986} with \cref{prop:charOCT}. Indeed, by examining the separable metrics given in \cite{Kalnins1986}, it can be shown that all separable metrics derived in \cite{Kalnins1986} have the form given by \cref{prop:charOCT}. Then the desired concircular tensor, $L$, is given by \cref{prop:charOCT}. For a space of constant curvature with arbitrary signature, it can be shown that the classification given by Kalnins and Miller can be generalized in such a way that the separable metrics still satisfy the hypothesis of \cref{prop:charOCT}. This generalization will be given in \cite{Rajaratnam2014d}.

\begin{proof}[\cref{thm:SCCconChKT}]
	Suppose inductively that this theorem holds for all separable webs in spaces of constant curvature of dimension $k < n$. The statement trivially holds when $k = 1$. We now show that the theorem holds when $\dim M = n$.

	Suppose $K$ is a ChKT defined on a space of constant curvature $M$ defining a separable web. Then let $L$ be a concircular tensor guaranteed by the KEM separation theorem.
	\begin{parts}
		\item If $L$ has simple eigenfunctions (i.e. is a Benenti tensor), then it follows that the separable web determined by $K$ is a KEM web.
		\item  Suppose $L$ has multidimensional eigenspace $D_{i}$ for $i = 1,...,l$; these must be Killing by \cref{prop:charOCT}. Thus each $D_{i}$ induces a foliation of spherical submanifolds of $M$. Then it follows by \cref{lem:SCCspherFol} that this is a foliation of spaces of constant curvature of lesser dimension. Suppose $N_{i}$ is an integral manifold of $D_{i}$. Then it follows from \cref{prop:restrictCKT} that $K$ restricts to a ChKT $\tilde{K}_{i}$ on $N_{i}$. Thus $\tilde{K}_{i}$ is a ChKT on a space of constant curvature $N_{i}$ which has dimension less than $n$. Hence by induction hypothesis, it follows that the separable web $\Ei_{i}$ associated with $\tilde{K}_{i}$ is a KEM web. Thus by definition it follows that the separable web associated with $K$ is a KEM web.
	\end{parts}
	
	The result then follows by induction on $n$.
\end{proof}

Motivated by \cref{thm:SCCconChKT}, in the next subsection we will use the results presented in \cref{sec:sepHJeqnWP} to give a recursive algorithm for separating potentials.

\subsection{The BEKM Separation Algorithm}

In this section we will present the Benenti-Eisenhart-Kalnins-Miller (BEKM) separation algorithm, which is named after the researchers who's work anticipated this algorithm \cite{Benenti2005a,Eisenhart1934,Kalnins1986a}. We fix a potential $V \in \F(M)$ and suppose $n = \dim M > 1$.

We first motivate the BEKM separation algorithm for the case when $M$ is a space of constant curvature. According to \cref{thm:HJosepII} a necessary condition for separability of $V$ is the existence of a ChKT $K$ satisfying the dKdV equation with $V$. Now \cref{thm:KEMsep} implies that there exists a non-trivial OCT $L$ which commutes as a linear operator with $K$. Hence the KBDT associated with $L$, say $K'$, is in the KS-space associated with $K$, thus by \cref{thm:HJosepII} $V$ must satisfy the dKdV equation with $K'$. This establishes the necessity of KBDTs for orthogonal separation in spaces of constant curvature. We use this fact and the theory on the separation of the Hamilton-Jacobi equation in warped products to obtain a recursive algorithm to find separable coordinates for $V$.

\begin{remark}
	The authors originally discovered the necessity of KBDTs for $\E^{n}$ and $\Si^{n}$ implicitly through Corollary 5.4 in \cite{Waksjo2003}. Indeed, according to the remarks following Equation~4.2 in \cite{Benenti2004}, the Bertrand-Darboux equations in \cite{Waksjo2003} are the dKdV equations generated by a KBDT. Hence Corollary 5.4 in \cite{Waksjo2003} implies the necessity of KBDTs for the special case of $\E^{n}$. Corollary 5.4 in \cite{Waksjo2003} also implies a similar statement for $\Si^{n}$. This explains the origin of the name Bertrand-Darboux in Killing-Bertrand-Darboux tensor and one of our initial reasons for working with CTs.
\end{remark}

Now we present the BEKM separation algorithm, so assume $M$ is an arbitrary pseudo-Riemannian manifold. Let $L$ denote the general concircular tensor on $(M,g)$ and $K := \tr{L} G - L$ be the KBDT generated by $L$. Now impose the condition:

\begin{equation} \label{eq:KBD}
	\d{(K \d V)} = 0
\end{equation}

\noindent which is called the \defn{kbd} equation. The above equation defines a system of linear equations in the unspecified parameters of $L$. Indeed, by \cref{thm:SCKTsFiniteDim}, the C-tensors form a finite-dimensional vector space. Since the KBDT is linearly related to $L$, it follows that the above equation defines a linear system. Furthermore by \cref{thm:SCKTsFiniteDim} the maximum number of unknowns in the above equation is $\frac{1}{2}(n+1)(n+2)$.

Suppose now that $K$ is a particular solution of the KBD equation and let $L$ be the associated C-tensor. We make the assumption that $L$ is an orthogonal tensor (which is always satisfied on a Riemannian manifold). Let $(E_{i})_{i=1}^{k}$ be the eigenspaces of $L$ and $(\lambda_{i})_{i=1}^{k}$ the corresponding eigenfunctions. We now classify such a solution:

\begin{parts}
		\item (k = 1, i.e. all the eigenfunctions coincide) \\
		In this case $L = c  G$ where $c := \lambda_{1} \in \R$, thus the associated KBDT, $K = c (n-1) G$ is the trivial solution of \cref{eq:KBD} and so the algorithm yields no information.
		\item (the eigenfunctions are simple) \\
		$K$ is a characteristic Killing tensor, then by Benenti's theorem (\cref{thm:HJosepII}), $V$ is separable in the web of the eigenspaces of $L$.
		\item (at least one eigenfunction is not simple) \\
		In this case, we enumerate the eigenspaces $D_{1},...,D_{l}$ with dimension greater than one. Since each $D_{i}$ is Killing by \cref{prop:charOCT}, the net formed by $D_{1},...,D_{l}$ together with $D_{0} := \bigcap\limits_{i=1}^{l} D_{i}^{\perp}$ is a WP-net. So fix $\bar{p} \in M$ and let $N = \prod_{i=0}^{l} N_{i}$ be a connected product manifold adapted to this net and passing through $\bar{p}$.
		
		If $D_{0} \neq 0$, then $K$ restricted to $D_{0}$ is characteristic by construction. Let $V_{i} := \tau_{i}^{*} V \in \F(N_{i})$ and suppose for each $i = 1,...,k$ there exists a ChKT $\tilde{K}_{i}$ on $N_{i}$ such that $\d (\tilde{K}_{i} \d V_{i}) = 0$.
			
		Then by \cref{thm:wpSOV}, $V$ is separable in the web formed by the simple eigenspaces of $L$ together with the lifts of the simple eigenspaces of $\tilde{K}_{1},...,\tilde{K}_{l}$.
\end{parts}

The algorithm can be applied recursively in the case $L$ has a non-simple eigenfunction. In the notation of case~3 one would have to apply the algorithm to each $N_{i}$ equipped with the induced metric for $i = 1,...,l$.

Now, some remarks are in order:

\begin{remark}
	In case~3 even if there are no ChKTs on the submanifolds $N_{i}$ which satisfy the dKdV equation with $V_{i}$, the Hamilton-Jacobi equation is partially separable.
\end{remark}
\begin{remark}
	Since the metric is always a solution of the KBD equation and because the KBD equation is linear in $K$, we always consider a solution of the KBD equation modulo multiplies of the metric.
\end{remark}

In the following example we will show how to use the theory just presented to show that the Calogero-Moser system is separable in cylindrical coordinates. It was originally shown to be separable in these coordinates by \citeauthor{Calogero1969} in \cite{Calogero1969}.

\begin{example}[Calogero-Moser system]
	The Calogero-Moser system is a natural Hamiltonian system with configuration manifold $\E^3$ given by the following potential in Cartesian coordinates $(q_{1},q_{2},q_{3})$:
	
	\begin{equation}
		V = (q_{1} - q_{2})^{-2} + (q_{2} - q_{3})^{-2} + (q_{1} - q_{3})^{-2}
	\end{equation}
	
	First note that the constant vector $d = \frac{1}{\sqrt{3}} [1,1,1]$ is a symmetry of $V$, i.e. $\lied{V}{d} = 0$. Hence we observe that the CT $L = d \odot d$ is a solution of the KBD equation associated with $V$. From \cref{ex:kemWeb} we know that a warped product manifold adapted to $L$ has the form $\E^1 \times \E^2$. One can choose Cartesian coordinates $(q_1',q_2',q_3')$ adapted to this product manifold, such that $V$ takes the form:
	
	\begin{equation}
		V = \frac{9(q_{3}'^2 + q_{2}'^2 )^2}{2 q_{2}'^2 (3 q_3'^2 - q_{2}'^2)^2}
	\end{equation}
	
	In this case $V$ naturally restricts to a potential on $\E^2$ with coordinates $(q_2',q_3')$. In $\E^2$ one can apply the BEKM separation algorithm to find that the only solution of the KBD equation (up to constant multiplies) is $L = r \odot r$ where $r$ is the dilatational vector field. Hence we conclude that $V$ is separable in cylindrical coordinates which are obtained by taking polar coordinates $(r \cos(\theta), r \sin(\theta) )$ on $\E^2$.
\end{example}

We won't give a proof here, but one can show that for the $n$-dimensional Calogero-Moser system the general solution to the KBD equation is

\begin{equation}
	L = c \, d \odot d + w \, d \odot r + m \, r \odot r
\end{equation}

\noindent where $c,w,m \in \R$ and $d = \frac{1}{\sqrt{n}} [1,\dotsc,1]$. One would have to apply the BEKM separation algorithm recursively when $n > 3$ in order to search for separable coordinates. In particular when $n = 3$, using the above solution one can determine that the Calogero-Moser system is separable in four additional coordinate systems. The details will appear in a subsequent article.

The following example illustrates how one can obtain ChKTs when an ignorable coordinate is present.
\begin{example}[Separation in Static space-times]
	A static space time is the product manifold $M = B \times_{\rho} \E_{1}^{1}$ equipped with warped product metric $g = \tilde{g} - \rho^2 \d t^2$ where $\tilde{g}$ is a Riemannian metric. By \cref{prop:wpKss}, $M$ admits a ChKT $K$ with timelike eigenvector field $\pderiv{}{t}$ iff there exists a ChKT $\tilde{K} \in S^2(B)$ satisfying: 
	
	\begin{equation}
		\d (\tilde{K} \d \rho^{-2}) = 0
	\end{equation}
	
	This observation is a special case of the connection between separation of potentials and extensions of KTs observed earlier via the Eisenhart metric. We note here that in order to find $\tilde{K}$, the BEKM separation algorithm can be applied on $B$ with $V := \rho^{-2}$. In particular if $B$ is a space of constant curvature, we will observe shortly that the BEKM separation algorithm gives a complete method for determining $\tilde{K}$ satisfying the above equation if it exists.
\end{example}

\paragraph{Completeness of the algorithm} It follows from the definition of the KEM web that if this algorithm is applied recursively then it will always test if the potential is separable in a KEM web. Thus it follows from \cref{thm:SCCconChKT} that this algorithm gives a complete test for separability in spaces of constant curvature. Although if one uses a ChKT not associated with a KEM web in case~3 of the algorithm, then one can test for separability against more general separable webs.

\paragraph{Practical Implementation} In \cite{Rajaratnam2014e} we will describe how to actually implement this algorithm in spaces of constant curvature. To do this, the only problems that remain are the classification of OCTs modulo the action of the isometry group, then obtaining the transformation to Cartesian coordinates for their associated webs and classifying warped product decompositions on these spaces.

This algorithm has been implemented concretely in Euclidean and spherical space by Waksjo and Wojciechowski in their solution \cite{Waksjo2003}. Their solution which was more classical, involved St\"{a}ckel theory and was based on the work of Kalnins-Miller \cite{Kalnins1986}. They made no use of Benenti's modern formulation of the separation of the Hamilton-Jacobi equation \cite{Benenti1997a} in terms of Killing tensors which is independent of St\"{a}ckel theory. 

Like the algorithm in \cite{Waksjo2003}, in spaces of constant curvature the BEKM separation algorithm reduces to a series of problems in linear algebra. Although for hyperbolic space and Minkowski space-time, one will have to deal with finding the Jordan canonical form of non-diagonalizable (constant) matrices.

\paragraph{Characterization of KEM webs} As we have mentioned above, the BEKM separation algorithm is a test for separability in KEM webs. Thus for the purposes of classifying orthogonal separation on a given manifold, one would want a more mathematically convenient definition of a KEM web then the natural one given earlier. First of all, we know that a KEM web is a separable web by \cref{prop:KEMisSep}. We also know that a KEM web is characterized by a hierarchy of orthogonal concircular tensors all diagonalized in adapted coordinates. Generalizing Crampin's observation \cite{Crampin2003}, by calculating the integrability conditions of the defining equation of a concircular tensor, one can show that the Riemann curvature tensor $R$ must satisfy the diagonal curvature condition in adapted coordinates, which is that $R_{ijik} = 0$ for $j \neq k$. Thus coordinates adapted to a KEM web are orthogonal separable coordinates and have diagonal curvature. In \cite{Rajaratnam2014d} we will solve for metrics that satisfy these conditions thereby determining if they characterize KEM webs. We will also use these ideas to prove the KEM separation theorem.

\section{Conclusion}

In this article we have introduced the notion of KEM webs and have shown that all separable webs in spaces of constant curvature are KEM webs. The remaining theory developed herein can be used for the study of these webs. These webs arise from the freedom that warped products give for constructing Killing tensors. \Cref{prop:wpKss} in \cref{sec:KTinWP} can be used to calculate the KS-space of a KEM web. The BEKM separation algorithm can be used to determine the separability of natural Hamiltonians in KEM webs. Furthermore, the theory developed in \cref{sec:KTinWP,sec:sepHJeqnWP} is applicable to any reducible separable web.

Concircular tensors are often studied with the added assumption of simple eigenfunctions, see \cite{Crampin2003,Benenti2005a} for example. Thus the results in this article can be seen as an application of these tensors when this assumption is weakened to point-wise diagonalization. Since such tensors appear in other areas such as geodesic equivalence \cite{Bolsinov2003} and cofactor systems (geodesically Hamiltonian systems) \cite{Lundmark2003,Crampin2001,Benenti2005a}, one can speculate that results in these areas can be generalized using the ideas presented in this article.

For strictly pseudo-Riemannian manifolds, one can attempt to extend the BEKM separation algorithm by including the case when the CT has complex eigenfunctions. Some relevant results in this area may be found in \cite{Degiovanni2007,Bolsinov2013}. Also non-orthogonal separation can occur as well \cite{Benenti1992,Benenti1997a}; this case is not considered in this article.

By definition of a KEM web, it follows that concircular tensors are invariant tensors characterizing these webs. Thus, according to \cref{thm:SCCconChKT}, every separable web in a space of constant curvature is characterized by a hierarchy of concircular tensors. As we mentioned earlier, this property allows for an algebraic calculation of the KS-space of these webs using Benenti's theory. Thus a related question arises: can one obtain invariant tensors for the webs associated with conformal separation which would allow for the algebraic calculation of the associated KS-space?

In \cite{Rajaratnam2014d} we will give an independent proof of the KEM separation theorem which is an important result within the theory. In \cite{Rajaratnam2014} we will show how to generalize Benenti's theory given in \cite{Benenti1993} and use the results presented in this article to explicitly obtain a basis for the KS-space for KEM webs. In \cite{Rajaratnam2014e} we will classify the concircular tensors in spaces of constant curvature thereby giving another classification of the separable webs in these spaces and enabling one to explicitly apply the BEKM separation algorithm. Finally, all the results presented in this article and the articles just listed can be found in thesis of the first author \cite{Rajaratnam2014}.

\section*{Acknowledgments}

We would like to express our appreciation to Dong Eui Chang for his continued interest in this work.  The research was supported in part by National Science and Engineering Research Council of Canada Discovery Grants (D.E.C. and R.G.M.). The first author would like to thank Spiro Karigiannis for reading his thesis \cite{Rajaratnam2014}, which contains the contents of this article.
\phantomsection
\addcontentsline{toc}{section}{References}
\printbibliography
\end{document}